\documentclass[11pt,a4paper]{article}
\usepackage{amsmath}
\usepackage{amsthm}
\usepackage{amssymb}
\usepackage{answers}
\usepackage{setspace}
\usepackage{graphicx}
\usepackage{enumitem}
\usepackage{multicol}
\usepackage{blkarray}
\usepackage{tikz}
\usetikzlibrary{automata, positioning, arrows}
\usepackage[colorlinks = true,
            linkcolor = blue,
            urlcolor  = blue,
            citecolor = blue,
            anchorcolor = blue]{hyperref}
\usepackage{mathrsfs}
\usepackage{lmodern}
\usepackage{amsmath,amsthm,amssymb}
\usepackage[T1]{fontenc}
\usepackage[noend]{algpseudocode}
\usepackage{comment}
\usepackage{textcomp}
\usepackage[ruled,vlined]{algorithm2e}

\newcommand{\Q}{\mathbb{Q}}

\newcommand{\C}{\mathbb{C}}
\newcommand{\R}{\mathbb{R}}
\newcommand{\F}{\mathbb{K}}

\newtheorem{theorem}{Theorem}[section]
\newtheorem{corollary}{Corollary}[theorem]
\newtheorem{lemma}[theorem]{Lemma}

\newtheorem{definition}[theorem]{Definition}
\newtheorem{proposition}[theorem]{Proposition}
\newtheorem{remark}[theorem]{Remark}

\newenvironment{example}[2][Example]{\begin{trivlist}
\item[\hskip \labelsep {\bfseries #1}]}{\end{trivlist}}

\usepackage[utf8]{inputenc}
\usepackage[english]{babel}
\usepackage{graphicx}
\graphicspath{{images/}{../images/}}
\usepackage{blindtext}
\newcommand{\noop}[1]{}
\usepackage{subfiles}
\title{Subfile Example}
\author{Team Learn Overleaf}
\date{ }

\begin{document}
\title{Black Box Absolute Reconstruction\\ for  Sums of Powers of Linear Forms }

\author{Pascal Koiran and Subhayan Saha\thanks{Univ Lyon, EnsL, UCBL, CNRS,  LIP, F-69342, LYON Cedex 07, France.
  Email: firstname.lastname@ens-lyon.fr.}}
\maketitle
\begin{abstract}
We study the decomposition of multivariate polynomials as sums of powers of linear forms. We give a randomized algorithm for the following problem: If a homogeneous polynomial $f \in K[x_1 , . . . , x_n]$ (where $K \subseteq \mathbb{C}$) of degree $d$ is given as a blackbox, decide whether it can be written as a linear combination of $d$-th powers of linearly independent complex linear forms. %We refer to this problem as the polynomial equivalence to some polynomial in $\mathcal{P}_d = \{\sum_{i =1 }^n \alpha_ix_i^d|\alpha_i \neq 0 \text{ for all } i\}$.
The main novel features of the algorithm are:
\begin{itemize}
    \item For $d = 3$, we improve by a factor of $n$ on the running time from the algorithm in \cite{koiran2020derandomization}. The price to be paid for this improvement though is that the algorithm now has two-sided error.
    \item For $d > 3$, we provide the first randomized blackbox algorithm for this problem that runs in time $\text{poly}(n,d)$ (in an algebraic model where only arithmetic operations and equality tests are allowed). Previous algorithms for this problem \cite{kayal11} as well as most of the existing reconstruction algorithms for other classes appeal to a polynomial factorization subroutine. This requires extraction of complex polynomial roots at unit cost and in standard models such as the unit-cost RAM or the Turing machine this approach does not yield polynomial time algorithms.
    \item For $d > 3$, when $f$ has rational coefficients (i.e. $K = \mathbb{Q}$), the running time of the blackbox algorithm is polynomial in $n,d$ and the maximal bit size of any coefficient of $f$. This yields the first algorithm for this problem over $\mathbb{C}$ with polynomial running time in the bit model of computation. 
\end{itemize}
These results are true even when we replace $\mathbb{C}$ by $\mathbb{R}$.
We view the problem as a tensor decomposition problem and use linear algebraic methods such as checking the simultaneous diagonalisability of the slices of a tensor. The number of such slices is exponential in $d$. But surprisingly, we show that after a random change of variables, computing just $3$ special slices is enough. We also show that our approach can be extended to the computation of the actual decomposition. %(this step relies on matrix diagonalisation which is not an algebraic step over $\mathbb{C}$).
In forthcoming work we plan to extend these results to overcomplete decompositions, i.e., decompositions in more than $n$ powers of linear forms.

\end{abstract}
\section{Introduction}

Lower bounds and polynomial identity testing are two fundamental problems about arithmetic circuits.
In this paper we consider another fundamental problem: arithmetic circuit reconstruction.
For an input polynomial $f$, typically given by a black box, the goal is to find the smallest circuit computing $f$ within some class $\cal C$ of arithmetic circuits. 
This problem can be divided in two subproblems: a decision problem (can $f$ be computed by a circuit of size $s$ from the class 
$\cal C$?) and the reconstruction problem proper (the actual construction of the smallest circuit for $f$).
In this paper we are interested in  {\em absolute reconstruction}, namely, in the case where $\cal C$ is a class of circuits over the field of complex numbers.  The name is borrowed from {\em absolute factorization}, a well-studied problem in computer algebra (see e.g. \cite{cheze05,ChezeLecerf07,gao03,shaker09}). 
Most of the existing reconstruction algorithms appeal to a polynomial factorization subroutine, 
see e.g.~\cite{GKP17ISSAC,GKP18,karnin09,kayal11,kayal18,kayal19,shpilka09}.
%\footnote{Some exceptions include...} 
This typically yields polynomial
time algorithms over finite fields or the field of rational numbers. However,  in standard models of computation such as the unit-cost RAM or the Turing machine %we do {\em not} obtain 
this approach does {\em not} yield polynomial time algorithms for absolute reconstruction. This is true even for the decision version of this problem. 
In the Turing machine model, the difficulty is as follows. We are given an input polynomial $f$, say with rational coefficients, and want to decide if there is a small circuit $C \in \cal C$ for $f$, where 
$C$ may have complex coefficients. 
After applying a polynomial factorization subroutine, a reconstruction algorithm will manipulate polynomials with coefficients
in a field extension of $\Q$. If this extension is of exponential degree,  the remainder of the algorithm will not
run in polynomial time. This point is explained in more detail in~\cite{koiran2020derandomization} on the example of a reconstruction algorithm 
due to Neeraj Kayal~\cite{kayal11}. 
One way out of this difficulty is to work in a model where polynomial roots can be extracted at unit cost, as suggested in a footnote of~\cite{garg19}. We will work instead in more standard models, namely, the Turing machine  model 
or the unit-cost RAM over~$\C$ with arithmetic operations only (an appropriate formalization
is provided by the { Blum-Shub-Smale} model of computation~\cite{BCSS,BSS89}).
Before presenting our results, we present the class of circuits studied in this paper.

\subsection{Sums of powers of linear forms} \label{sec:powersums}

Let $f(x_1,\ldots,x_n)$ be a homogeneous polynomial of degree $d$.
%, also called a {\em degree $d$ form.}  
In this paper we study decompositions of the type:
\begin{equation} \label{eq:linind}
  f(x_1,\ldots,x_n)=\sum_{i=1}^r l_i(x_1,\ldots,x_n)^d
\end{equation}
where the $l_i$ are linear forms. Such a decomposition is sometimes called
a Waring decomposition, or a symmetric tensor decomposition. The smallest possible value of $r$ is the symmetric tensor rank of $f$, and it is NP-hard to compute already for $d=3$~\cite{shitov16}. One can nevertheless obtain polynomial time algorithms by restricting to a constant value of $r$~\cite{bhargava21}. In this paper we assume instead that the linear forms $l_i$ are linearly independent (hence $r \leq n$). This setting was already studied by Kayal~\cite{kayal11}. 
It turns out that such a decomposition is unique when it exists, up to a permutation of the~$l_i$ and multiplications by $d$-th roots of unity. This follows for instance from Kruskal's uniqueness theorem. For a more elementary proof, see~\cite[Corollary 5.1]{kayal11} and~\cite[Section 3.1]{koiran2020derandomization}.

Under this assumption of linear independence, the case $r=n$ is of particular interest. In this case, $f$ is {\em equivalent} to
the sum of $d$-th powers polynomial 
\begin{equation} \label{eq:pd}
P_d(x)=x_1^d+x_2^d+\cdots+x_n^d
\end{equation}
in the sense that $f(x)=P_d(Ax)$ where $A$ is invertible.
A test of equivalence to $P_d$ was provided in~\cite{kayal11}.
The resulting algorithm provably runs in polynomial time over the field of rational numbers, but this is not the case over $\C$ due to the appeal to polynomial factorization. The first equivalence test to $P_d$ running in polynomial time over the field complex numbers was given in~\cite{koiran2020derandomization} for $d=3$. We will extend this result to arbitrary degree in this paper.
In the general case $r \leq n$ we can first compute the number of essential variables of $f$~\cite{carlini06,kayal11}. 
Then we can do a change of variables to obtain a polynomial depending only 
on its first $r$ variables~\cite[Theorem 4.1]{kayal11},
and conclude with a test of equivalence to $P_r$ (see \cite[Proposition 44]{koiran2020derandomization} for details).

Equivalence and reconstruction algorithms over $\Q$ are number-theoretic in nature in the sense 
that their behavior is highly sensitive to number-theoretic properties of the coefficients of the input polynomial. This point is clearly illustrated by an example from~\cite{koiran2020derandomization}:
 \begin{example}  %\label{ex:norat}
 
  Consider the rational  polynomial 
$$f(x_1,x_2)=(x_1+\sqrt{2}x_2)^3+(x_1-\sqrt{2}x_2)^3 = 2x_1^3+12x_1x_2^2.$$
This polynomial is equivalent to $P_3(x_1,x_2)=x_1^3+x_2^3$ over $\R$ and $\C$
but not over $\Q$.
\end{example}
 By contrast, equivalence and reconstruction algorithms over $\R$ and $\C$ are of a more geometric nature.
\subsection{Sums of cubes} \label{sec:cubes}

For $d=3$, the first test of equivalence to $P_d$ running in polynomial time over $\C$ and over $\R$ was given in~\cite{koiran2020derandomization}. There, the problem was treated as a tensor decomposition problem which was then solved by methods from linear algebra. We briefly outline this approach since the present paper improves on it and extends it to higher degree. Let $f \in \F[x_1,\ldots,x_n]$ be the input polynomial, where $\F$ is the field of real or complex numbers. We can form with the coefficients of $f$ a symmetric tensor\footnote{Recall that a tensor of order $d$ is symmetric it is invariant under all $d!$ permutations of its indices.} of order three $T=(T_{ijk})_{1 \leq i,j,k \leq n}$ so that
$$f(x_1,\ldots,x_n)=\sum_{i,j,k=1}^n T_{ijk} x_i x_j x_k.$$
This tensor can be cut into $n$ slices $T_1,\ldots,T_n$ where $T_k = (T_{ijk})_{1 \leq i,j \leq n}$. Each slice is a symmetric matrix of size $n$. By abuse of language we also say that $T_1,\ldots,T_n$ are the slices of $f$.
The equivalence test to $P_3$ proposed in~\cite{koiran2020derandomization} works as follows.
\begin{enumerate}
\item On input $f \in \F[x_1,\ldots,x_n]$, pick a random matrix $R \in M_n(\F)$ and set $h(x)=f(Rx)$.
\item Let $T_1,\ldots,T_n$ be the slices of $h$. If $T_1$ is singular, reject. Otherwise, compute $T'_1=T_1^{-1}$.
\item If  the matrices $T'_1T_k$ commute and are all diagonalizable over $\F$, accept. Otherwise, reject.
\end{enumerate}
This simple randomized algorithm has one sided error: it can fail (with low probability) only when $f$ is equivalent to $P_3$.
Its analysis is based on the following characterization \cite[Section 3.2]{koiran2020derandomization}: 
\begin{theorem}\label{thm:socchar}
A degree $3$ homogeneous polynomial $f \in \F[x_1,...,x_n]$ is equivalent to $P_3$ iff its slices $T_1,...,T_n$ span a non-singular matrix space and the slices are simultaneously diagonalisable by congruence, i.e., there exists an invertible matrix $Q \in M_n(\F)$ such that $Q^TT_iQ$ is diagonal for all $i \in [n]$.
\end{theorem}
\subsection{Connection to Tensor Decomposition}\label{subsec:tensorconn}
Using the relation between tensors and polynomials, we can see that a homogeneous degree-$d$ polynomial $f \in \mathbb{K}[x_1,...,x_n]$ can be written as a sum of $d$-th of linear forms over $\mathbb{K}$ if and only if there exist $v_i \in \mathbb{K}$ such that the corresponding symmetric tensor $T_f$ can be decomposed as $T_f = \sum_{i} v_i^{\otimes d}$. This is often referred to as the tensor decomposition problem for the given tensor~$T$. 
\par
Most tensor decompositions algorithms are numerical.  Without any attempt at exhaustivity, one may cite the ALS method \cite{Kolda2009TensorDA} (which lacks a good complexity analysis), tensor power iteration \cite{JMLR:v15:anandkumar14b} (for orthogonal tensor decomposition) or Jennrich's algorithm~\cite{Harshman1970FoundationsOT,moitra2018algorithmic} for ordinary tensors.
Unlike the above algorithm from~\cite{koiran2020derandomization}, these numerical algorithms do not provide any decision procedure.
The algebraic algorithm from~\cite{koiran2020derandomization} seems closest in spirit to Jennrich's: they both rely on simultaneous diagonalization and on linear independence assumptions on the vectors involved in the tensor decomposition.
Algorithms for symmetric tensor decomposition can be found in the algebraic literature, see e.g. \cite{brachat10,bernardi11}.
These two papers do not provide any complexity analysis for their algorithms.

\subsection{Results and methods} \label{sec:results}

Our main contributions are as follows. Recall that $P_d$ is the sum of $d$-th powers polynomials~(\ref{eq:pd}),
and let us assume that the input $f \in \C[x_1,\ldots,x_n]$ is a homogeneous polynomial of degree $d$.
\begin{itemize}
\item[(i)] For $d=3$, we improve by a factor of~$n$ on the running time of the test of equivalence to $P_3$ from~\cite{koiran2020derandomization} presented in Section~\ref{sec:cubes}. 
The price to be paid for this improvement is that the algorithm now has two-sided error.

\item[(ii)] For $d>3$, we provide the first blackbox algorithm for equivalence to~$P_d$ with running time polynomial in $n$ and $d$,
in an algebraic model where only arithmetic operations and equality tests are allowed (i.e., computation of polynomial roots are {\em not allowed}). 
\item[(iii)] For $d>3$, when $f$ has rational coefficients this blackbox algorithm runs in polynomial time in the bit model of computation. More precisely, the running time is polynomial in $n,d$ and the maximal bit size of any coefficient of $f$.
This yields the first test of equivalence to $P_d$ over $\C$ with polynomial running time in the bit model of computation.
\end{itemize}
As outlined in Section~\ref{sec:powersums}, these results have application to decomposition into sums of powers of linearly 
independent linear forms over $\C$. Namely, we can decide whether the input polynomial admits such a decomposition, and if it 
does we can compute the number of terms $r$ in such a decomposition. The resulting algorithm runs in polynomial time
in the algebraic model of computation, as in item (ii) above; when the input has rational coefficients it runs in polynomial time in the bit model of computation, as in (iii) (refer to Appendix \ref{app:dthpowers} for a detailed complexity analysis). This is the first algorithm with these properties. 
It can be viewed as an algebraic, high order, black box version of Jennrich's algorithm. 
\par
Using the relation to tensor decomposition problem mentioned in Section \ref{subsec:tensorconn}, if an order $d$-tensor $T \in K^{n \times ... \times n}$ is given as a blackbox, we give an algorithm that runs in time $\text{poly}(n,d)$ to check if there exist linearly independent vectors $v_i \in \mathbb{K}^n$ such that $T = \sum_{i = 1}^t \alpha_i v_i^{\otimes d}$ for some $t \leq n$. Note here that $K \subseteq \mathbb{C}$ and $\mathbb{K} = \mathbb{C} \text{ or } \mathbb{R}$.
\par
As an intermediate result, we obtain a new randomized algorithm for checking that $k$ input matrices commute 
(see Lemma~\ref{lem:comm} towards the end of this section).
\par
Finally, we show that our linear algebraic approach can be extended to the computation of the actual  decomposition.
For instance, when $f \in \C[x_1,\ldots,x_n]$ is equivalent to $P_d$, 
we can compute an invertible matrix $A$ such that $f(x)=P_d(Ax)$.
We emphasize that for this result we must step out of our usual algebraic model, and allow the computation of polynomial
roots. The matrix $A$ is indeed not computable from $f$ with arithmetic operations only, as shown %for instance 
by the example in Section~\ref{sec:powersums}. We therefore obtain an alternative to the algorithm from~\cite{kayal11} for the computation of~$A$. That algorithm relies on multivariate polynomial factorization, whereas our algorithm relies on matrix diagonalization
(this is not an algebraic task since diagonalizing a matrix requires the computation of its eigenvalues).

{\bf Real versus complex field.} For $\mathbb{K} = \R$ and even degree there is obviously a difference between sums of $d$-th powers of linear forms and linear combinations of $d$-th powers. In this paper we wish to allow arbitrary
linear combinations. For this reason, in the treatment of the high order case ($d>3$) we are not  interested in equivalence to $P_d$ only. Instead, we would like to know whether the input is equivalent to some polynomial of the form 
$\sum_{i=1}^n \alpha_i x_i^d$ with $\alpha_i \neq 0$ for all $i$. We denote by ${\cal P}_d$ this class of polynomials
(one could even assume that $\alpha_i = \pm 1$ for all $i$). At first reading, there is no harm in assuming that 
$\mathbb{K} = \C$. In this case, one can assume without loss of generality that $\alpha_i=1$ for all $i$. 
For $\mathbb{K}=\R$, having to deal with the whole of~${\cal P}_d$ slightly complicates notations, but the proofs are not
 significantly more complicated than for $\mathbb{K}=\C$. For this reason, in all of our results we give a  unified treatment of
 the two cases $\mathbb{K} = \C$ and $\mathbb{K} = \R$.

{\bf Methods.} In order to extend the approach of Section~\ref{sec:cubes} to higher order, we associate to a homogeneous
polynomial of degree $d$ the (unique) symmetric tensor $T$ of order $d$ such that
$$f(x_1,\ldots,x_n)=\sum_{i_1,\ldots,i_d=1}^n T_{i_1...i_d} x_{i_1} x_{i_2} \ldots x_{i_d}.$$
We show in Section \ref{sec:equivd} that Theorem~\ref{thm:socchar} can be generalized as follows:
\begin{theorem}\label{thm:Pdcharintro}
A degree $d$ homogeneous polynomial   $f \in \mathbb{C}[x_1,\ldots,x_n]$ is equivalent to %$P_d = \sum_i \alpha_i x_i^d$ 
$P_d = \sum_{i=1}^n x_i^d$ 
%such that $\alpha_i \neq 0$
 if and only if its slices %$T_{11...1},...,T_{nn...n}$ 
 span a nonsingular matrix space and the slices are simultaneously diagonalizable by congruence, i.e., there exists an invertible matrix  $Q \in M_n(\C)$ %$Q \in M_n(\mathbb{K})$ 
 such that %the matrices $Q^T T_{i_1...i_{d-2}} Q$ are diagonal for all $i_1...i_{d-2} \in [n]$.
for every slice $S$ of~$f$, the matrix $Q^TSQ$ is diagonal.
\end{theorem}
This characterization is satisfactory from a purely structural point of view, but not from an algorithmic point of view because the number of slices of a tensor of order $d$ is exponential in $d$.
Recall indeed that a  tensor of size $n$ and order $d$ has $\frac{d(d-1)}{2} n^{d-2}$ slices:  a slice is obtained by
fixing the values of $d-2$ indices. The tensors encountered in this paper are all symmetric since they originate from
homogeneous polynomials. Taking the symmetry constraints into consideration reduces the number of distinct slices 
to ${n+d-3 \choose d-2}$ at most: this is the number of multisets of size $d-2$ in a set of $n$ elements, or equivalently
the number of monomials of degree $d-2$ in the variables $x_1,\ldots,x_n$.
This number remains much too large to reach our goal of a complexity polynomial in $n$ and $d$. 
This problem has a surprisingly simple solution: our equivalence algorithm needs to work with 3 slices only!
This is true already for $d=3$, and is the reason why we can save a factor of~$n$ compared to the algorithm of Section~\ref{sec:cubes}.
More precisely, we can replace the loop at line 3 of that algorithm by the following test: check that $T'_1 T_2$ is diagonalizable, and commutes with $T'_1 T_3$ (recall that $T'_1=T_1^{-1}$). It may be surprising at first sight that we can work with the first 3 slices only of a tensor 
with $n$ slices. To give some plausibility to this claim, note that $T_1,T_2,T_3$ are not slices of the input $f$, but slices
of the polynomial $h(x)=f(Rx)$ obtained by a random change of variables. As a result, each slice of $h$  contains some
information on {\em all} of the $n$ slices of $f$.
The algorithm for order $d>3$ is of a similar flavor, but one must be careful in the choice of the 3 slices from~$h$.
Our algorithms are therefore quite simple (and the equivalence algorithm for $d=3$ is even somewhat simpler
than the algorithm from Section~\ref{sec:cubes}); but their analysis is not so simple and 
forms the bulk of this paper.

As a byproduct of our analysis of the degree 3 case, we obtain a randomized algorithm for testing the
commutativity of a family of matrices $A_1,\ldots,A_k$. The naive algorithm for this would check that $A_iA_j = A_j A_i$ 
for all $i \neq j$. Instead, we propose to test the commutativity of two random linear combinations of the $A_i$.
The resulting algorithm has one sided-error, and its probability of error can be bounded as follows:
\begin{lemma} \label{lem:comm}
Let $A_1,...,A_k \in M_n(\mathbb{K})$. We take two random linear combinations $A_{\alpha} = \sum_{i \in [k]} \alpha_i A_i$ and $A_{\beta} = \sum_{i \in [k]} \beta_i A_i$, where the $\alpha_i$ and $\beta_i$ are picked independently and uniformly at random from a finite set $S \subset \mathbb{K}$. If $\{A_i\}_{i \in [k]}$ is not a commuting family, then the two matrices $A_{\alpha}, A_{\beta}$ commute with probability at most $\frac{2}{|S|}$.
\end{lemma}
The resulting algorithm is extremely simple and natural, but we could not find in the literature on commutativity testing.
Commutativity testing has been studied in particular in the setting of black box groups, in the classical~\cite{Pak12} and quantum models~\cite{magniez07}.
Pak's algorithm~\cite{Pak12} is based on the computation of random subproducts of the $A_i$. 
In its instantiation to matrix groups~\cite[Theorem~1.5]{Pak12}, Pak suggests as a speedup to apply Freivald's technique~\cite{freivalds79} for the verification
of matrix products. This can be done in the same manner for Lemma~\ref{lem:comm}.
We stress that Pak's algorithm applies only to groups rather than semigroups; in particular, for the application to commutativity of matrices this means that the $M_i$ must be invertible.\footnote{Pak's result is definitely stated only for groups, and it appears that its correctness proof actually uses the invertibility hypothesis.} Note that there is no such assumption in Lemma~\ref{lem:comm};
compared to~\cite{Pak12} we therefore obtain a randomized algorithm for testing {\em matrix semigroup commutativity}.
We also note that the idea of testing commutativity on random linear combinations is akin to the general technique for 
the verification of identities in~\cite{schulman00}. 
However, in the case of commutativity testing that paper does not obtain any improvement over the trivial deterministic algorithm (see Theorem~3.1 in~\cite{schulman00}).
In order to analyze the higher order case $d >3$, we will derive an appropriate generalization of Lemma~\ref{lem:comm} (to matrices satisfying certain symmetry properties).

\subsection{Organization of the paper}
In Section \ref{sec:fastercubes}, we present a faster algorithm for equivalence to sum of cubes. We give a detailed complexity analysis of our algorithm in Appendix \ref{app:cubes} and compare it to that of \cite{koiran2020derandomization}.
In Section \ref{sec:dthpowers}, we extend our ideas for the degree-$3$ case to the arbitrary degree-$d$ case and give an algorithm for equivalence to sum of $d$-th powers (Algorithm \ref{alg:pd}). In fact our algorithm can test if the input polynomial is equivalent to some linear combination of $d$-th powers (As explained in Section \ref{sec:results}, these notions are different over $\mathbb{R}$ when $d$ is even). In Appendix \ref{app:pdalgebraic}, we give a detailed complexity analysis of Algorithm \ref{alg:pd}.
%we show that if the polynomial is given as a blackbox, Algorithm \ref{alg:pd} makes $\text{poly}(n,d)$ many calls to the blackbox and performs $\text{poly}(n,d)$ many arithmetic operations to decide equivalence to some polynomial in $\mathcal{P}_d$. 
In Appendix \ref{sec:bitcomplexity}, we show that when the input polynomial has rational coefficients, Algorithm~\ref{alg:pd} runs in polynomial time in the bit model of computation, as well. In Section \ref{sec:minvars}, we give an algorithm to check whether the input polynomial can be decomposed into a linear combination of $d$-th powers of at most $n$ many linearly independent linear forms.  In Appendix \ref{app:minvars}, we compute the number of blackbox calls and arithmetic operations performed by this algorithm.
In Section \ref{sec:reconstruction}, we show how we can modify our decision algorithm to give an algorithm that actually computes the linear forms and their corresponding coefficients. 
 
\subsection{Notations}\label{sec:notations}

We work in a field $\mathbb{K}$ which may be the field of real numbers or the field of complex numbers.
Some of our intermediate results (in particular, Lemma~\ref{lem:comm}) apply to other fields as well.
We denote by $\mathbb{K}[x_1,\ldots,x_n]_d$ the space of homogeneous polynomials of degree $d$ in $n$ variables
with coefficients in $\mathbb{K}$. A homogeneous polynomial of degree $d$ is also called a {\em degree-$d$ form}.
We denote by $P_d$ the polynomial $\sum_{i=1}^n x_i^d$, and we say that a degree $d$ form $f(x_1,\ldots,x_n)$ is equivalent to a sum of $d$-th powers if it is equivalent to $P_d$, i.e., if $f(x)=P_d(Ax)$ for some invertible matrix $A$.
{ More generally, we denote by ${\cal P}_d$ the set of polynomials of the form $\sum_{i=1}^n \alpha_i x_i^d$
with $\alpha_i \neq 0$ for all $i$. As explained in Section~\ref{sec:results}, for $\mathbb{K} = \R$ we are not only interested in equivalence to $P_d$: we would like to know whether the input is equivalent to one of the elements of  ${\cal P}_d$.}

We denote by $M_n(\mathbb{K})$ the space of square matrices of size $n$ with entries from $\mathbb{K}$. 
We denote by $\omega$ a feasible exponent for matrix multiplication, i.e., we assume that
two matrices of   $M_n(\mathbb{K})$ can be multiplied with $O(n^{\omega})$ arithmetic operations in $\mathbb{K}$.
\par
We denote by $M(d)$ the number of arithmetic operations required for multiplication of two polynomials of degree $\leq d$ and we will often refer to the $O(d\log d \log\log d)$ bounds given by \cite{Schonhage1971} for polynomial multiplication to give concrete bounds for our algorithms.
\par
Throughout the paper, we will choose the entries $r_{ij}$ of a matrix $R$ independently and uniformly at random from a finite set $S \subset {\mathbb{K}}$. When we calculate the probability of some event $E$ over the random choice of the $r_{ij}$, by abuse of notation instead of $\Pr_{r_{11},\ldots,r_{nn}\in S}[E]$ we simply write  $\Pr_{R \in S}[E]$.

\section{Faster algorithm for sums of cubes}\label{sec:fastercubes}

In this section we present our fast algorithm for checking whether an input polynomial $f(x_1,\ldots,x_n)$ is equivalent to 
$P_3=x_1^3+\cdots+x_n^3$ (see Algorithm~\ref{algo:cubes} below). As explained in Section~\ref{sec:powersums}, this means that $f(x)=P_3(Ax)$ for some invertible matrix $A$. In Section~\ref{sec:cubes} we saw that a degree 3 form in $n$ variables can be viewed as an order 3 tensor, which we can be cut into $n$ slices. All of our decomposition algorithms build on this approach.
%Our algorithm builds on this approach, and proceeds as follows:

%\begin{figure}
\begin{algorithm} \label{algo:cubes}
\SetAlgoLined
 \textbf{Input:} A degree-3 homogeneous polynomial $f$ \\
% \While{While condition}{
Let $R \in M_n(\mathbb{K})$ be a matrix such that its entries $r_{ij}$ are picked uniformly and independently at random from a finite set $S$ and set $h(x) = f(Rx)$ \\
Let $T_1, T_2, T_3$ be the first 3 slices of $h$. \\
\eIf{$T_1$ is singular}{reject}{compute $T_1' = T_1^{-1}$ \\
   \eIf{ $T_1'T_2$ and $T_1'T_3 $ commute and $T_1'T_2$ is diagonalisable over $\mathbb{K}$}{
   accept
   }{
   reject
  }
  }
 \caption{Randomized algorithm to check equivalence to $P_3$}
\end{algorithm}
%\end{figure}

Recall from Section~\ref{sec:cubes} that the equivalence algorithm from~\cite{koiran2020derandomization} needs to check that
the $n$ matrices $T_1'T_k$ commute and are diagonalisable, where $T_1,\ldots,T_n$ denote the slices of $h(x)=f(Rx)$.
Algorithm~\ref{algo:cubes} is faster because it only checks that $T_1'T_2$ and $T_1'T_3 $ commute and that $T_1'T_2$ is diagonalisable. We do a detailed complexity analysis of the two algorithms in Appendix~\ref{app:cubes}. It reveals that the cost
of the diagonalisability tests dominates the cost of the commutativity tests for both algorithms.
Since we have replaced $n$  diagonalisability tests by a single test, it follows that Algorithm~\ref{algo:cubes} is faster by a factor of~$n$. More precisely, we show that the algorithm from~\cite{koiran2020derandomization} performs $O(n^{\omega+2})$ arithmetic operations when $\mathbb{K} = \C$, but Algorithm~\ref{algo:cubes}  performs only $O(n^{\omega+1})$ arithmetic operations.

The remainder of this section is devoted to a correctness proof for Algorithm~\ref{algo:cubes}, including an analysis of the probability of error. Our main result about this algorithm is as follows.
\begin{theorem} \label{th:equiv3}
If an input $f \in \mathbb{K}[x_1,...,x_n]_3$ is not equivalent to a sum of cubes, then $f$ is rejected by the algorithm with high probability over the choice of the random matrix $R$. More precisely, if the entries $r_{i,j}$ are chosen uniformly and independently at random from a finite set $S \subseteq \mathbb{K}$  %with $|S|>2n$, 
then the input will be rejected with probability $\geq 1-\frac{2}{|S|}$.

Conversely, if $f$ is equivalent to a sum of n cubes then $f$ will be accepted
with high probability over the choice of the random matrix $R$. More precisely, if the entries $r_{i,j}$ are chosen uniformly and independently at random from a set $S \subseteq \mathbb{K}$, then the input will be accepted with probability $\geq 1-\frac{2n}{|S|}$.
\end{theorem}
The second part of Theorem~\ref{th:equiv3} is the easier one, and it already follows from~\cite{koiran2020derandomization}.
Indeed, the same probability  of error $2n/|S|$ was already given for the randomized equivalence algorithm of~\cite{koiran2020derandomization}, and any input accepted by that algorithm is also accepted by our faster equivalence algorithm. Nevertheless, we give a self-contained proof of this error bound in Section~\ref{sec:positive} as a preparation 
toward the case of higher degree.

One of the reasons why the analysis is simpler for positive inputs is that there is only one way for a polynomial to be equivalent to $P_3$: its slices must satisfy all the properties of Theorem~\ref{thm:main} (at the end of Section~\ref{sec:equiv3}). By contrast, if a polynomial is not 
equivalent to $P_3$ this can happen in several ways depending on which property fails.
We analyze failure of commutativity in Section~\ref{sec:failcom} and failure of diagonalisability in  Section~\ref{sec:faildiag}.
Then we tie  everything together  in Section~\ref{sec:neg}.

\subsection{Characterization of equivalence to $P_3$} \label{sec:equiv3}

Toward the proof of Theorem~\ref{th:equiv3} we need some results from~~\cite{koiran2020derandomization}, which we recall in this section. We also give a complement in Theorem~\ref{thm:diagonalisable subspace}. First, let us recall how the slices of a polynomial evolve under a linear change of variables.
\begin{theorem}\label{thm:P3structural}
Let $g$ be a degree-$3$ form with slices $S_1,...,S_n$ and let $f(x) = g(Ax)$. The slices
$ T_1,...,T_n$ of $f$ are given by the formula:
\begin{align*}
    T_k = A^T D_k A
\end{align*}
 where $D_k = \sum_{i=1} a_{i,k} S_i$ and the $a_{i,k}$ are the entries of A.
In particular, if %$g = P_3(Ax)$ 
$g=\sum_{i=1}^n \alpha_i x_i^3$ we have $D_k = \text{diag }(\alpha_1 a_{1,k},...,\alpha_n a_{n,k} )$.
\end{theorem}
In Theorem~\ref{thm:socchar} we gave a characterization of equivalence to $P_3$ based on simultaneous diagonalisation
by congruence. This characterization follows from Theorem~\ref{thm:P3structural} and the next lemma.
See~\cite[Section 3.2]{koiran2020derandomization} for more details on Theorem~\ref{thm:P3structural}, Lemma~\ref{lem:subspace} and the connection to Theorem~\ref{thm:socchar}.
\begin{lemma}\label{lem:subspace}
Let $f$ be a degree 3 homogeneous polynomial such that $f(x) = P_3(Ax)$ for some non-singular $A$. Let $\mathcal{U}$ and $\mathcal{V}$ be the subspaces of $M_n(\mathbb{K})$ spanned by slices of $f$ and $P_3$ respectively. Then the subspace $\mathcal{V}$ is the space of diagonal matrices and $\mathcal{U}$ is a non-singular subspace, i.e., it is not made of singular matrices only. 
\end{lemma}

%moved earlier
%\begin{theorem}\label{thm:socchar}
%A degree $3$-form $f \in \mathbb{K}[x_1,...,x_n]$ is equivalent to a sum of $n$ cubes iff its slices $T_1,...,T_n$ span a non-singular matrix space and the slices are simultaneously diagonalisable by congruence i.e there exists an invertible matrix $Q \in M_n(\mathbb{K})$ such that $Q^TT_iQ$ is diagonal for all $i \in [n]$.
%\end{theorem}
Instead of diagonalisation by congruence, it is convenient to work with the more familiar notion of diagonalisation by similarity, 
where an invertible matrix $A$ acts by $S \mapsto A^{-1} SA$ instead of $A^TSA$. We collect the necessary material 
in the remainder of this section (and we refer to diagonalisation by similarity simply as {\em diagonalisation}).

The two following properties play a fundamental role throughout the paper. 
\begin{definition} %1
%Let $A_1,...,A_k$ be complex symmetric matrices of size $n$ such that the subspace $\mathcal{V}$ spanned by these matrices is non-singular.
Let $\mathcal{V}$ be a non-singular space of matrices.
\begin{itemize}
    \item We say that $\mathcal{V}$ satisfies the \textbf{Commutativity Property} if there exists an invertible matrix 
    $A \in \mathcal V$ such that $A^{-1} \mathcal V$ is a commuting subspace.
    \item We say that $\mathcal{V}$ satisfies the \textbf{Diagonalisability Property} if there exists an invertible matrix $B \in \mathcal V$ such that all the matrices in the space  $B^{-1} \mathcal V$ are diagonalisable.
\end{itemize}
\end{definition}

The next result can be found in~\cite[Section 2.2]{koiran2020derandomization}.

\begin{theorem}\label{thm:commuting subspace}
Let $\mathcal{V}$ be a non-singular subspace of matrices of $M_n(\mathbb{K})$. The following properties are equivalent.
\begin{itemize}
    \item $\mathcal{V}$ satisfies the {commutativity property}.
    \item For all non-singular matrices $A \in \mathcal{V}$,  $A^{-1}\mathcal{V}$ is a commuting subspace.
\end{itemize}
\end{theorem}
\begin{remark} \label{rem:diagprop}
Let $\mathcal{V}$ be a non-singular subspace of matrices which satisfies the commutativity and diagonalisability properties. There exists an invertible matrix $B \in \mathcal{V}$  and an invertible matrix $R$ which diagonalizes simultaneously all of $B^{-1} \mathcal V$ (i.e., $R^{-1}MR$ is diagonal for all $M \in B^{-1} \mathcal V$).
\end{remark}
\begin{proof}
Pick an invertible matrix $B \in \mathcal V$ such that ${\mathcal W}=B^{-1} \mathcal V$ is a space of diagonalizable matrices. 
By Theorem~\ref{thm:commuting subspace}, $\mathcal W$ is a commuting subspace.
It is well known that a finite collection of matrices is simultaneously diagonalisable if and only if they commute, and
each matrix in the collection is diagonalisable. We conclude by applying this result to a basis of $\mathcal W$ (any matrix $R$ which diagonalises a basis will diagonalise all of $\mathcal W$).
\end{proof}
We now give an analogue of Theorem~\ref{thm:commuting subspace} for the diagonalisability property.
\begin{theorem}\label{thm:diagonalisable subspace}
Let $\mathcal{V}$ be a non-singular subspace of matrices which satisfies the commutativity property. The following properties are equivalent:
\begin{itemize}
    \item $\mathcal{V}$ satisfies the {diagonalisability property}.
    \item For all non-singular matrices $A \in \mathcal{V}$,  the matrices in $A^{-1} \mathcal V$ are simultaneously diagonalisable.
\end{itemize}
\end{theorem}
\begin{proof}
Suppose that $\mathcal{V}$ satisfies the {diagonalisability property}. By the previous remark, we already know 
that there exists {\em some} invertible matrix $B \in \mathcal V$ such that the matrices in $B^{-1} \mathcal V$ are simultaneously diagonalisable by an invertible matrix $R$. We need to establish the same property for an arbitrary invertible matrix $A \in \mathcal V$.
For any $M \in \mathcal V$, $A^{-1}M = (B^{-1}A)^{-1} (B^{-1}M)$. Hence $A^{-1}M$ is diagonalised by $R$ since this matrix diagonalises both matrices
$B^{-1}A$ and $B^{-1} M$. % are simultaneously diagonalisable.
Since $R$ is independent of the choice of $M\in  \mathcal V$, we have shown that the matrices in $A^{-1} \mathcal V$ are simultaneously diagonalisable.
%It follows like in the proof of Remark~\ref{rem:diagprop} that the matrices in $A^{-1} \mathcal V$ are simultaneously diagonalisable because they commute are are individually diagonalisable.
\end{proof}
The importance of the commutativity and diagonalisability properties stems from the fact that they provide a characterization 
of simultaneous diagonalisation by congruence, which in turn (as we have seen in Theorem~\ref{thm:socchar}) provides a
characterization of equivalence to $P_3$:
\begin{theorem}\label{thm:main}
Let $A_1 , . . . , A_k \in M_n(\mathbb{K})$ and
assume that the subspace $\mathcal{V}$ spanned by these matrices is non-singular. There are diagonal matrices $\Lambda_i$ and a non-singular matrix $R \in M_n(\mathbb{K})$ such that
$A_i = R \Lambda_i R^T$ for all $i \in [k]$ if and only if $\mathcal{V}$ satisfies the {Commutativity property} and the {Diagonalisability property}.
\end{theorem}
For a proof, see~\cite[Section 2.2]{koiran2020derandomization} for $\mathbb{K} = \C$ 
and~\cite[Section 2.3]{koiran2020derandomization} for $\mathbb{K} = \R$.

\subsection{Analysis for positive inputs} \label{sec:positive}

In this section we analyze the behavior of Algorithm~\ref{algo:cubes} on inputs that are equivalent to $P_3$. First, we recall the Schwartz-Zippel lemma which we will be using throughout the paper.
\begin{lemma}[\cite{DL78}\cite{Zip79}\cite{Sch80}]\label{lem:SZ}
Let $P \in \mathbb{K}[x_1,...,x_n]$ be a non-zero polynomial of total degree $d \geq 0$ over a field $\mathbb{K}$. Let $S$ be a finite subset of $\mathbb{K}$ and let $r_1,...,r_n$ be picked uniformly and independently at random from a finite set $S$. Then
\begin{align*}
    \text{Pr}_{r_1,...,r_n \in S}[P(r_1,...,r_n) = 0] \leq \frac{d}{|S|}.
\end{align*}
\end{lemma}
\begin{lemma}\label{lem:t1invertible}
Let $f$ be a degree-3 form with slices $S_1,...,S_n$ such that the subspace $\mathcal{V}$ spanned by the slices is non-singular.
Let $h(x) = f(Rx)$ where  the entries $r_{i,j}$ are chosen uniformly and independently at random from a finite set $S \subseteq \mathbb{K}$. Let $T_1,...,T_n$ be the slices of $h$. Then
\begin{align*}
    \text{Pr}_{R \in S}[T_1 \text{ is invertible}] \geq 1- \frac{2n}{|S|}.
\end{align*}
\end{lemma}
\begin{proof}
We can obtain the slices $T_k$ of $h$ from the slices $S_k$ of $f$ using Theorem~\ref{thm:P3structural} namely, we have $T_k = R^TD_kR$ where $D_k = \sum_{i \in [n]} r_{i,k}S_i$ and the $r_{i,k}$ are the entries of $R$.
 \newline
 Therefore $T_1$ is invertible iff $R$ and $D_1$ are invertible.
 Applying the Schwartz-Zippel lemma to $\det(R)$ shows that $R$ is singular with probability at most ${n}/{|S|}$.
 We will see that  $D_1$ is singular also with probability at most ${n}/{|S|}$;  the lemma then follows 
 from the union bound.
% Now $R$ is not invertible, iff $\text{det}(R) = 0$.
% Using Schwartz-Zippel we get that
% \begin{align*}
%     \text{Pr}_{R \in S}[\text{det}(R) = 0] \leq \frac{n}{|S|}
 %\end{align*}
 Matrix $D_1$ is not invertible iff $\text{det}(D_1) = 0$.
 %\newline
 Since $D_1 = \sum_{i \in [n]} r_{i,1}S_i$ $\det(D_1) \in \mathbb{K}[r_{1,1},...,r_{n,1}]$ and $\text{deg}(\text{det}(D_1)) \leq n$. Since, $\mathcal{V}$ is non-singular, there exists some choice of $\alpha = (\alpha_1,...,\alpha_n)$, such that $S = \sum_{i \in [n]}\alpha_iS_i$ is invertible. Hence $\det(D_1)$ is not identically zero, and it follows again from Schwartz-Zippel lemma that this polynomial vanishes with probability at most $\frac{n}{|S|}$.
 \end{proof}
 
 \begin{lemma}\label{lem:opp}
 Given $A \in M_n(\mathbb{K})$, let $T_1,...,T_n$ be the slices of $h(x)=P_3(Ax)$. If $T_1$ is invertible, define $T'_1 = (T_1)^{-1}$. Then $T'_1T_2$ commutes with $T'_1T_3$, and $T'_1T_2$ is diagonalisable.
 % \begin{itemize}
  %   \item $T'_1T_2$ and $T'_1T_3$ commute
 %   \item $T'_1T_2$ is diagonalisable.
 %\end{itemize} 
\end{lemma}
\begin{proof}
By Theorem \ref{thm:P3structural}, 
$$T_k = A^T\text{diag}(A_{1k},...,A_{nk})A = A^TD_1A.$$
%\begin{align*}
%    T_1 &= R^T(\text{diag}(r_1,...,r_{n1}))R = R^TD_1R \\
 %   T_2 &= R^T(\text{diag}(r_{12},...,r_{n2}))R = R^TD_2R \\
%    T_3 &= R^T(\text{diag}(r_{13},...,r_{n3}))R = R^TD_3R
%\end{align*}
If $T_1$ is invertible, the same is true of $A$ and $D_1$.
The inverse $(D_1)^{-1}$ is diagonal like $D_1$, hence $(D_1)^{-1}D_2$ and $(D_1)^{-1}D_3$ are both diagonal as well and must therefore commute. Now,
\begin{align*}
    T'_1T_2T'_1T_3 &= A^{-1}((D_1)^{-1}D_2(D_1)^{-1}D_3)A \\
    &= A^{-1}((D_1)^{-1}D_3(D_1)^{-1}D_2)A \\
    &= T'_1T_3T'_1T_2.
\end{align*}
\par
Finally,  $T'_1T_2 = A^{-1}((D_1)^{-1}D_2)A$ so this matrix diagonalisable.
\end{proof}
In the above lemma we have essentially reproved the easier half of Theorem~\ref{thm:main}. 
 We are now in position to prove the easier half of Theorem~\ref{th:equiv3}.
 \begin{proposition} \label{prop:equiv3}
If an input $f \in \mathbb{K}[x_1,...,x_n]_3$ is equivalent to a sum of~$n$ cubes then $f$ will be accepted
by Algorithm~\ref{algo:cubes} with high probability over the choice of the random matrix $R$. More precisely, if the entries $r_{i,j}$ are chosen uniformly and independently at random from a set $S \subseteq \mathbb{K}$, then $f$ will be accepted with probability at least $1-\frac{2n}{|S|}$.
\end{proposition}
\begin{proof}
Suppose that $f(x)=P_3(Bx)$ for some invertible matrix $B$. By Lemma~\ref{lem:subspace}, the space spanned by the  slices of $f$ is nonsingular. We can therefore apply Lemma~\ref{lem:t1invertible}: the first slice $T_1$ of $h(x)=f(Rx)$ 
is invertible with probability at least $1-\frac{2n}{|S|}$.
Moreover, when $T_1$ is invertible Lemma~\ref{lem:opp} shows that $f$ will always be accepted (we can apply this lemma to 
$h$ since $h(x)=P_3(BRx)$).
\end{proof}

\subsection{Failure of commutativity} \label{sec:failcom}

In this section we first give the proof of Lemma~\ref{lem:comm}. This is required for the analysis of Algorithm~\ref{algo:cubes}, and moreover this simple lemma yields a new randomized algorithm for commutativity testing as explained in Section~\ref{sec:results}. We restate the lemma here for the reader's convenience:
\begin{lemma}
Let $A_1,...,A_k \in M_n(\mathbb{K})$. We take two random linear combinations $A_{\alpha} = \sum_{i \in [k]} \alpha_i A_i$ and $A_{\beta} = \sum_{i \in [k]} \beta_i A_i$, where the $\alpha_i$ and $\beta_i$ are picked independently and uniformly at random from a finite set $S \subset \mathbb{K}$. If $\{A_i\}_{i \in [k]}$ is not a commuting family, then the two matrices $A_{\alpha}, A_{\beta}$ commute with probability at most $\frac{2}{|S|}$.
\end{lemma}
\begin{proof}
%Consider therefore a non-commuting family $A_1,\ldots,A_k \in M_n(\mathbb{K})$, and take the two random linear combinations $A_{\alpha} = \sum_{i \in [k]} \alpha_i A_i$ and $A_{\beta} = \sum_{j \in [k]} \beta_j A_j$,
We want to bound the probability of error, i.e.,
$\Pr_{\alpha,\beta}[A_{\alpha},A_{\beta} \text{ commute}].$
Let us define 
\begin{align*}
    P_{\text{comm}}(\alpha,\beta) &= A_{\alpha}A_{\beta} - A_{\beta}A_{\alpha} \\
    &= \sum_{i,j \in [k]} \alpha_i\beta_j(A_iA_j - A_jA_i).
\end{align*}
By construction, $A_{\alpha}$ commutes with $A_{\beta}$ if and only if
%\begin{equation}
  %      \label{eqn:comm}
        $P_{\text{comm}}(\alpha,\beta) = 0.$
%\end{equation}
%\newline
Since $\{A_i\}_{i \in [k]}$ is not a commuting family, there exists $i,j \in [n]$ such that % $A_i$ and $A_j$ don't commute. From definition of commutativity, 
$A_iA_j - A_jA_i \neq 0$. Hence there exists some entry $(r,s)$ such that
\begin{equation}\label{eq:noncomm}
    (A_iA_j - A_jA_i)_{r,s} \neq 0
\end{equation}
\newline
Let us define 
%\begin{align*}
    $P^{r,s}_{\text{comm}}(\alpha,\beta) = (A_{\alpha}A_{\beta} - A_{\beta}A_{\alpha})_{r,s}$.
%\end{align*}
From (\ref{eq:noncomm}) we have
%\begin{equation}\label{eq:noncommid}
   $$ P^{r,s}_{\text{comm}}(e_i,e_j) \neq 0$$
%\end{equation}
where $e_i$ is the vector with a 1 at the $i$-th position and 0's elsewhere. In particular, $P^{r,s}_{\text{comm}}$ is not identically zero. 
Since $\text{deg}(P^{r,s}_{\text{comm}}) \leq 2 $, it follows from the 
Schwartz-Zippel lemma that
\begin{align*}
    \text{Pr}_{\alpha,\beta \in S}[P^{r,s}_{\text{comm}}(\alpha,\beta) = 0] \leq \frac{2}{|S|}
\end{align*}
and the same upper bound applies to 
   $\text{Pr}_{\alpha,\beta \in S}[P_{\text{comm}}(\alpha,\beta) = 0]$.\
\end{proof}
The next result relies on the above lemma. Theorem \ref{thm:algcomm} gives us a way to analyze the case when the slices of 
the input polynomial fail to satisfy the commutativity property (recall that this property is relevant due to Theorem~\ref{thm:main}):
\begin{theorem}\label{thm:algcomm}
Let $f \in \mathbb{K}[x_1,...,x_n]_3$ be a degree 3 form such that the subspace $\mathcal{V}$  spanned by its $n$ slices %$S_1,...,S_n$ 
is non-singular and does not satisfy the {commutativity property}.
%\newline
Let $h(x) = f(Rx)$ where the entries $r_{i,j}$ of $R$ are chosen uniformly and independently at random from a finite set $S \subset \mathbb{K}$.% with $|S|>2n$. 
Let $T_1,...,T_n$ be the slices of $h$. 
%\newline
If $T_1$ is invertible, define $T_1' = T_1^{-1}$. Then 
\begin{align*}
    \text{Pr}[T_{1} \text{ is invertible and } T_{1}'T_{2}, T_{1}'T_{3} \text{ commute} ] \leq \frac{2}{|S|}.
\end{align*}
\end{theorem}
\begin{proof}
%We begin by defining the following event
%\begin{align*}
%    E_1 &:= \text{Event such that } T'_{1}T_{2}, T'_{1}T_{3} \text{ commute }
%\end{align*}
By Theorem~\ref{thm:P3structural} we know that $T_k = R^T(\sum_{i=1}^n r_{i,k}S_i)R$
where $S_1,\ldots,S_n$ are the slices of $f$.
%\newline
Let us define $D_1=\sum_{i=1}^n r_{i,1}S_i$.
Then we have:
\begin{align*}
    T_1'T_2 &= R^{-1}(D_1)^{-1}R^{-T}R^T(\sum_{i=1}^n r_{i,2}S_i)R \\
    %& = R^{-1}(D_1)^{-1}(\sum_{i=1}^n r_{i,2}S_i)R \\
    &= R^{-1}(\sum_{i=1}^n r_{i,2}D_1^{-1}S_i)R.
\end{align*}
Similarly,
%\begin{align*}
 $\displaystyle   T_1'T_3 = R^{-1}(\sum_{i=1}^n r_{i,3}D_1^{-1}S_i)R.$
%\end{align*}
%So now $T_1'T_2$ and $T_1'T_3$ commute implies that
%\begin{align*}
%    T_1'T_2T_1'T_3 - T_1'T_3T_1'T_2 = 0 
%\end{align*}
So $T'_1T_2$ commutes with $T'_1T_3$ iff $R$ is invertible and $\sum_{i=1}^n r_{i,2}D_1^{-1}S_i$ commutes with 
$\sum_{i=1}^n r_{i,3}D_1^{-1}S_i$.
%\begin{comment}
%    & R^{-1}(\sum_{i=1}^n r_{i,2}D_1^{-1}S_i)RR^{-1}(\sum_{i=1}^n r_{i,3}D_1^{-1}S_i)R - R^{-1}(\sum_{i=1}^n r_{i,3}D_1^{-1}S_i)RR^{-1}(\sum_{i=1}^n r_{i,2}D_1^{-1}S_i)R = 0 \\
%    & R^{-1}(\sum_{i,j = 1}^n r_{i,2}r_{j,3} (D_1^{-1}S_iD_1^{-1}S_j-D_1^{-1}S_jD_1^{-1}S_i)R = 0 \\
%    & \sum_{i,j = 1}^n r_{i,2}r_{j,3} (D_1^{-1}S_iD_1^{-1}S_j-D_1^{-1}S_jD_1^{-1}S_i) = 0 
%\end{comment}
Let  $E_1$ be the event that $T'_{1}T_{2}$ commutes with $T'_{1}T_{3}$, and let $E'_1$ be the event that $\sum_{i=1}^n r_{i,2}D_1^{-1}S_i$ commutes with 
$\sum_{i=1}^n r_{i,3}D_1^{-1}S_i$.
  Let $E_2$ be the event that $\{(D_{1})^{-1}S_{i}\}_{i \in [n]}$  is not a commuting family.
Since $\mathcal{V}$ does not satisfy the commutativity property, $(D_{1})^{-1}\mathcal{V}$ is not a commuting subspace if $D_{1}$ is invertible.
%\newline
Hence the event that $D_1$ is invertible is the same as $E_2$. 
%\newline
Setting $A_{i} = (D_{1})^{-1}S_{i}$, $\alpha_{i} = r_{i,2}$, $\beta_i = r_{i,3}$ in Lemma \ref{lem:comm} we obtain
%\begin{equation}
  $$  \text{Pr}_{R \in S}\big[E'_1\big|E_2\big] \leq \frac{2}{|S|}.$$
%\end{equation}
Note here that $D_{1}$ depends only on the random variables $r_{i,1}$ for all $i \in [n]$ and therefore is independent of $r_{k,2}$ and $r_{l,3}$ for all $k,l \in [n]$, because we assume that the entries of $R$ are all picked uniformly and independently at random.
%\newline
%\textbf{Note:} Here we use the notation $\text{Pr}_{R \in S}$ to denote the probability when $r_{i,j}$'s are chosen uniformly and independently at random from $S$.
%\newline

Now we know that $T_{1}$ is invertible iff $R$ and $D_{1}$ are invertible.
Let  $E_3$ be the event  that  $T_{1}$  is invertible, and $E_4$  the event that $R$ is invertible.
We have $E_3 = E_2 \cap E_4$, % and    $ \text{Pr}[E_3] \geq 1 -{2n}/{|S|}$ by Lemma \ref{lem:t1invertible}. We 
and we have seen that $E_1=E'_1 \cap E_4$.
The probability of error can finally be bounded as follows:
$$ \text{Pr}_{R \in S}[E_1 \cap E_3] = \text{Pr}_{R \in S}[E'_1\cap E_2 \cap E_4]
     \leq {\text{Pr}_{R \in S}[E'_1| E_2]} \leq 2/|S|.$$
\end{proof}

\subsection{Failure of diagonalisability} \label{sec:faildiag}

Theorem \ref{thm:algcomm} gives us a way to analyze the case when the slices of 
the input polynomial fail to satisfy the commutativity property. With the results in the present section we will be able to analyze the case where the commutativity property is satisfied, but the diagonalisability property fails (recall that these properties
are relevant due to Theorem~\ref{thm:main}).
\begin{proposition}\label{prop:diagon}
Let $\mathcal{U} \subseteq M_n(\mathbb{K})$ be a commuting subspace of matrices. We define 
\begin{align*}
    \mathcal{M} := \Big\{M \Big| M \text{ is diagonalisable and } M \in \mathcal{U}\Big\}.
\end{align*}
Then $\mathcal{M}$ is a linear subspace of $\mathcal{U}$.
%\par
In particular, if there exists $A \in \mathcal{U}$ such that $A$ is not diagonalisable then $\mathcal{M}$ is a proper linear subspace of 
$\mathcal{U}$.
\end{proposition}
\begin{proof}

$\mathcal{M}$ is trivially closed under multiplication by scalars.
%\newline
Let $M,N \in \mathcal{M}$. These two matrices are diagonalisable by definition of $\mathcal{M}$, and they commute 
since $\mathcal{M} \subseteq \mathcal{U}$. Hence they are simultaneously diagonalisable. Thus $\mathcal{M}$ is closed under addition as well, which implies that it is a linear subspace of~$\mathcal{U}$.
%\par
%For the next case, let us assume on the contrary that $\text{dim}(\mathcal{M}) = \text{dim}(\mathcal{U})$. Also, $\mathcal{M} \subseteq \mathcal{U}$ where $\mathcal{M}$ is a vector space, then $\mathcal{M} = \mathcal{U}$. That implies that $A \in \mathcal{M}$, which is a contradiction.
\end{proof}
\begin{corollary}\label{corr:diag}
Let $\{A_i\}_{i \in [n]}$  be a commuting family of matrices such that $A_i$ is not diagonalisable for at least one index $i \in [n]$. Let $S \subset \mathbb{K}^n$ be a finite set. Then  $D = \sum_{i} \alpha_iA_i$ is diagonalisable with probability at most ${1}/{|S|}$ when $\alpha_1,...,\alpha_n$ are chosen uniformly and independently at random from $S$.
\end{corollary}
\begin{proof}
We define $\mathcal{U} = \text{span}\{A_1,...,A_n\}$ and
%\newline
%We define 
\begin{align*}
    \mathcal{M} := \Big\{M \Big| M \text{ is diagonalisable and } M \in \mathcal{U}\Big\}.
\end{align*}
So the probability of error is 
%\begin{equation}
%\label{eq:diagprob}
   $ \text{Pr}_{\bar{\alpha} \in S}\Big[D \in \mathcal{M}\Big].$
%\end{equation}
By Proposition \ref{prop:diagon} and the hypothesis that there exists $A_i \in \mathcal{U} \setminus \mathcal{M}$, $\mathcal{M}$ is a proper linear subspace of~$\mathcal{U}$. So $\mathcal{M}$ is an intersection of hyperplanes. Since $A_i \not\in \mathcal{M}$, there exists  a linear form $l_{\mathcal{M}}(X)$ corresponding to a hyperplane  such that $l_{\mathcal{M}}(M) = 0$ for all $M \in \mathcal{M}$ and $l_{\mathcal{M}}(A_i) \neq 0$. This gives us that $l_{\mathcal{M}} \not\equiv 0$. 
%Also, $\text{deg}(l_{\mathcal{M}}) = 1$ since it's the equation corresponding to a hyperplane. 
%We know that $D$ is diagonalisable iff $l_{\mathcal{M}}(D) = 0$
We know that if $D$ is diagonalisable then $l_{\mathcal{M}}(D) = 0$.
By the  Schwartz-Zippel Lemma %and Equation (\ref{eq:diagprob}), we get that 
the probability of error satisfies:
\begin{align*}
    \text{Pr}_{\bar{\alpha} \in S}\Big[D \in \mathcal{M}\Big] { \leq} \text{Pr}_{\bar{\alpha} \in S}\Big[l_{\mathcal{M}}(D) = 0\Big] \leq \frac{1}{|S|}
\end{align*}
since $\text{deg}(l_{\mathcal{M}}) = 1$.
\end{proof}

The last result of this section is an analogue of Theorem \ref{thm:algcomm} for the diagonalisability property.
\begin{theorem}\label{thm:diagonalisability}
Let $f \in \mathbb{K}[x_1,...,x_n]_3$ be a degree 3 form such that the subspace $\mathcal{V}$  spanned by its $n$ slices %$S_1,...,S_n$ 
is non-singular, satisfies the {commutativity property} but does not satisfy the diagonalisability property.
Let $h(x) = f(Rx)$ where the entries $r_{i,j}$ of $R$ are chosen uniformly and independently at random from a finite set $S \subset \mathbb{K}$. % with $|S|>2n$. 
Let $T_1,...,T_n$ be the slices of $h$. 
If $T_1$ is invertible, define $T_1' = T_1^{-1}$. Then
\begin{align*}
    \text{Pr}_{R \in S}[ T_{1} \text{ is invertible and } T_{1}'T_{2} \text{ is diagonalisable}] \leq \frac{1}{|S|}.
\end{align*}
\end{theorem}
\begin{proof}
As in the proof of Theorem \ref{thm:algcomm} we have
\begin{align*}
    T_1'T_2 = R^{-1}(\sum_{i=1}^n r_{i,2}D_1^{-1}S_i)R
\end{align*}
where $D_1=\sum_{i=1}^n r_{i,1}S_i$.
So $T_1'T_2$ is diagonalisable iff $R$ is invertible and $M = \sum_{j \in [n]}r_{j,2}D_1^{-1}S_j$ is diagonalisable. 
We denote by $E_1$ be the event that $T'_{1}T_{2}$ is diagonalisable, and by $E'_1$ the event that $M$ is diagonalisable. 

Let $E_2$ be the event that $\{(D_{1})^{-1}S_{i}\}_{i \in [n]}$ is a commuting family, but there exists $i \in [n]$ such 
that  $(D_{1})^{-1}S_{i}$ is not diagonalisable.
%\end{align*}
Since $\mathcal{V}$ satisfies the {commutativity property} and does not satisfy the {diagonalisability property}, 
by Theorem \ref{thm:diagonalisable subspace} the event that $D_1$ is invertible is the same event as $E_2$.

Setting $A_{i} = (D_{1})^{-1}S_{i}$ and  $\alpha_{i} = r_{i,2}$ in Corollary \ref{corr:diag}, we obtain
%\begin{equation}
    $$\text{Pr}_{R \in S}\big[E'_1\big|E_2\big] \leq \frac{1}{|S|}.$$
%\end{equation}
{Note} here that $D_{1}$ depends only on the random variables $r_{i,1}$ for all $i \in [n]$ and therefore is independent of $r_{k,2}$ for all $k \in [n]$, because we assume that the entries of $R$ are all picked uniformly and independently at random.

Now we know that $T_{1}$ is invertible iff $R$ and $D_{1}$ is invertible.
Let  $E_3$ be the event  that  $T_{1}$  is invertible, and $E_4$  the event that $R$ is invertible.
We have $E_3 = E_2 \cap E_4$, 
%and $\text{Pr}[E_3] \geq 1 -\frac{2n}{|S|}$ by Lemma \ref{lem:t1invertible}.
and we have seen that $E_1=E'_1 \cap E_4$.
The probability of error can finally be bounded as follows:
\begin{align*}
 %   \text{Pr}_{R \in S}[E_1|E_3] = \frac{\text{Pr}_{R \in S}[E_1\cap E_3]}{\text{Pr}_{R \in S}[E_3]} &= \frac{\text{Pr}_{R \in S}[E'_1\cap E_2 \cap E_4]}{\text{Pr}_{R \in S}[E_3]} \\
  \text{Pr}_{R \in S}[E_1 \cap E_3] = \text{Pr}_{R \in S}[E'_1\cap E_2 \cap E_4]
     \leq {\text{Pr}_{R \in S}[E'_1| E_2]} 
     \leq {\frac{1}{|S|}}.
\end{align*}
\end{proof}

\subsection{Analysis for negative inputs} \label{sec:neg}

In this section we complete the proof of Theorem~\ref{th:equiv3}. The case of positive inputs was treated in Section~\ref{sec:positive}. It therefore remains to prove the following result.
\begin{theorem}
If an input $f \in \mathbb{K}[x_1,...,x_n]_3$ is not equivalent to a sum of cubes, then $f$ is rejected by Algorithm~\ref{algo:cubes} with high probability over the choice of the random matrix $R$. More precisely, if the entries $r_{i,j}$ are chosen uniformly and independently at random from a finite set $S \subseteq \mathbb{K}$  %with $|S|>2n$, 
then the input will be rejected with probability at least $1-\frac{2}{|S|}$.
\end{theorem}
\begin{proof}
Let $S_1,...,S_n$ be the slices of $f$ and $\mathcal{V} =  \text{span}\{S_1,...,S_n\}$.
From Theorem \ref{thm:socchar} and Theorem \ref{thm:main}, we know that if $f \not\sim P_3$ there are three disjoint cases to consider:
\begin{itemize}
    \item[(i)] $\mathcal{V}$ is a singular subspace of matrices.
    \item[(ii)] $\mathcal{V}$ is a non-singular subspace and does not satisfy the {commutativity property}.
    \item[(iii)] $\mathcal{V}$ is a non-singular subspace, satisfies the {commutativity property} but does not satisfy the {diagonalisability property}.
\end{itemize}
 We will upper bound the probability of error in each case.
 In case (i), $T_1 = \sum_{j \in [n]} r_{1,j}S_j \in \mathcal{V}$ is always singular for any choice of the $r_{1,j}$.
 So $f$ is rejected by the algorithm with probability 1 in this case.
  In case (ii)   we can  upper bound the probability of error as follows:
\begin{align*}
        &\text{Pr}_{R \in S}[f \text{ is accepted by the algorithm}] \\
        &=\text{Pr}_{R \in S}[T_1 \text{ is invertible, }T_1'T_2, T_1'T_3 \text{ commute, } T_1'T_2 \text{ is diagonalisable}] \\
        &\leq \text{Pr}_{R \in S}[T_1 \text{ is invertible, }T_1'T_2, T_1'T_3 \text{ commute}]. %\\
        %&\leq \text{Pr}_{R \in S}[T_1'T_2, T_1'T_3 \text{ commute }|T_1 \text{ is invertible}].
    \end{align*}
    By Theorem \ref{thm:algcomm}, this occurs with probability $2/|S|$ at most.
     In case (iii) we have the following bound on the probability of error:
     \begin{align*}
        &\text{Pr}_{R \in S}[f \text{ is accepted by the algorithm}] \\
        &=\text{Pr}_{R \in S}[T_1 \text{ is invertible, }T_1'T_2, T_1'T_3 \text{ commute, } T_1'T_2 \text{ is diagonalisable}] \\
        &\leq \text{Pr}_{R \in S}[T_1 \text{ is invertible, }T_1'T_2 \text{ is diagonalisable}].  %\\
        %&\leq \text{Pr}_{R \in S}[T_1'T_2 \text{ is diagonalisable}|T_1 \text{ is invertible}].
    \end{align*}
  By Theorem \ref{thm:diagonalisability} this occurs with probability $1/|S|$ at most.
Therefore, in all three cases  the algorithm rejects $f$ with probability at least $1- \frac{2}{|S|}$.
\end{proof}

%\section{Old Theorems}
%\subfile{oldtheorems}
%\section{Notations}
%\subfile{notations}
\newpage
%{ \section{Remainder of the paper starts here}}
%\section{A commutativity lemma}
%\subfile{commutativity}
%\newpage
%\section{A diagonalisability lemma}
%\subfile{diagonalisability}
%\newpage
%\section{Efficient randomized equivalence algorithm for $P_3$}
%\subsection{The Algorithm}
%\subfile{randomized}
%\subsection{Some necessary lemmas}
%\subfile{necessarytheorems}
%\subsection{Theorem for commutativity}
%\subfile{randcomm}
%\subsection{Theorem for diagonalisability}
%\subfile{randdiag}
%\subsection{Theorem required for converse}
%\subfile{randopp}
%\subsection{Proof of Correctness}
%\subfile{randproof}
%\subsection{Complexity of the algorithm}
%\subfile{complexity}
%\newpage
%\section{Reproving old theorems for $P_4$}
%\subfile{p4reprove}
%\section{Randomized Algorithm for $P_4$}
%\subsection{The Algorithm}
%\subfile{p4algorithm}
%\subsection{Some necessary lemmas}
%\subfile{p4necessarytheorems}
%\subsection{Theorem for commutativity}
%\subfile{p4randcomm}
%\subsection{Theorem for diagonalisability}
%\subfile{p4randdiag}
%\subsection{Theorems required for converse}
%\subfile{p4randopp}
%\subsection{Proof of Correctness}
%\subfile{p4randproof}
%\subfile{p4complexity}
\section{Equivalence to a linear combination of $d$-th powers}\label{sec:dthpowers}
We can associate to a symmetric tensor $T$ of order $d$ the homogeneous polynomial
\begin{align*}
    f(x_1,...,x_n) = \sum_{i_1,...,i_d \in [n]} T_{i_1...i_d} x_{i_1}...x_{i_d}.
\end{align*}
\newline
This correspondence is bijective, and the symmetric
tensor associated to a homogeneous polynomial $f$ can be obtained from the
relation
\begin{align*}
    T_{i_1...i_d}  = \frac{1}{d!}\frac{\partial^d f}{\partial x_{i_1}...\partial x_{i_d}}.
\end{align*}
The $(i_1,...,i_{d-2})$-th slice of $T$ is the symmetric matrix $T_{i_1...i_{d-2}}$ with entries $(T_{i_1...i_{d-2}})_{i_{d-1},i_d} = T_{i_1...i_d}$
\subsection{The Algorithm}\label{sec:algod}
Recall from Section \ref{sec:notations}, we denote by $\mathcal{P}_d$, the set of polynomials of the form $\sum_{i=1}^n \alpha_ix_i^d$ with $\alpha_i \neq 0$ for all $i \in [n]$. In this section we present a poly-time algorithm for checking whether an input degree $d$ form in $n$ variables $f$ is equivalent to some polynomial in $\mathcal{P}_d$(see Algorithm \ref{alg:pd} below). This means that $f(x) = P_d(Ax)$ for some $P_d \in \mathcal{P}_d$ such that $A$ is invertible. 
\par
Recall from Section \ref{sec:fastercubes}, that the equivalence algorithm for sum of cubes needs to check if $T_1'T_2$ commutes with $T_1'T_3$ and if $T_1'T_2$ is diagonalisable, where $T_1,...,T_n$ are the slices of $h(x) = f(Rx)$. Now, we prove a surprising fact that even for the higher degree cases, checking commutativity of 2 matrices and the diagonalisability of 1 matrix is enough to check equivalence to sum of linear combination of $d$-th powers. Let $\{T_{i_1,...,i_{d-2}}\}_{i_1,...,i_{d-2} \in [n]}$ be the slices of $h(x) = f(Rx)$. We denote by $T_{\bar{i}}$, the corresponding slice $T_{i...i}$.
Algorithm \ref{alg:pd} checks if $T_{\bar{1}}'T_{\bar{2}}$ commutes with $T_{\bar{1}}'T_{\bar{3}}$ and if $T_{\bar{1}}'T_{\bar{3}}$ is diagonalisable.
\par
Interestingly though, arbitrary slices of a degree-$d$ polynomial are hard to compute. These particular slices are special because they can be computed using small number of calls to the blackbox and in small number of arithmetic operations (due to the fact that they are essentially repeated partial derivatives with respect to a single variable) and hence, help us give a polynomial time algorithm. More precisely, we show that if the polynomial is given as a blackbox, the algorithm requires only $O(n^2d)$ calls to the blackbox and $O(n^2M(d)\log d+ n^{\omega+1})$ many arithmetic operations. We do a detailed complexity analysis of this algorithm in Appendix \ref{app:dthpowers}.
\begin{algorithm}[H]\label{alg:pd}
\SetAlgoLined
%\KwResult{Write here the result }
 \textbf{Input:} A degree-d homogeneous polynomial $f$
 \\
% \While{While condition}{
Let $R \in M_n(\mathbb{K})$ be a matrix such that its entries $r_{ij}$ are picked uniformly and independently at random from a finite set $S$ and set $h(x) = f(Rx)$.  \\
Let $\{T_{i_1...i_{d-2}}\}_{i_1...i_{d-2} \in [n]}$ be the slices of $h$. \\
We compute the slices $T_{\bar{1}}$,$T_{\bar{2}}$,$T_{\bar{3}}$. \\
\eIf{$T_{\bar{1}}$ is singular}{reject}{compute $T_{\bar{1}}' = (T_{\bar{1}})^{-1}$ \\
   \eIf{ $T_{\bar{1}}'T_{\bar{2}}$ and $T_{\bar{1}}'T_{\bar{3}} $ commute and $T_{\bar{1}}'T_{\bar{2}}$ is diagonalisable over $\mathbb{K}$}{
   accept
   }{
   reject
  }
  }
 \caption{Randomized algorithm to check polynomial equivalence to $\mathcal{P}_d$}
\end{algorithm}
The remainder of this section is devoted to a correctness proof for Algorithm \ref{alg:pd}, including an analysis of the probability of error. Our main result
about this algorithm is as follows:
\begin{theorem}\label{thm:pdproof}
If an input $f \in \mathbb{F}[x_1,...,x_n]_d$ is not equivalent to some polynomial $P_d \in \mathcal{P}_d$, then $f$ is rejected by the algorithm with high probability over the choice of the random matrix $R$. More precisely, if the entries $r_{i,j}$ of $R$ are chosen uniformly and independently at random from a finite set $S \subseteq \mathbb{K}$, then the input will be rejected with probability $\geq (1-\frac{2(d-2)}{|S|})$.
\newline
Conversely, if $f$ is equivalent to some polynomial $P_d \in \mathcal{P}_d$, then $f$ will be accepted with high probability over the choice of the random matrix $R$. More precisely, if the entries $r_{i,j}$ are chosen uniformly and independently at random from a finite set $S \subseteq \mathbb{K}$, then the input will be accepted with probability $\geq (1-\frac{n(d-1)}{|S|})$.
\end{theorem}
The proof structure of this theorem follows the one of Theorem \ref{th:equiv3}. In Section~\ref{sec:positived}, we give a proof of the second part of theorem i.e. the behavior of Algorithm \ref{alg:pd} on the positive inputs. Here we require a stronger property of the subspace spanned by the slices of these positive inputs. For this we define the notion of "weak singularity" in Section \ref{sec:equivd} and prove an equivalence result related to it. On the negative inputs i.e if a polynomial is not equivalent to some polynomial in $\mathcal{P}_d$, this can again happen in several ways depending on which property fails. We analyze the failure of commutativity in Section \ref{sec:failcomd} and failure of diagonalisability in Section \ref{sec:faildiagd}. Then we collect everything together and prove the first part of the theorem in Section \ref{sec:negd}.

\subsection{Characterisation of equivalence to $\mathcal{P}_d$}\label{sec:equivd}
First, we show how the slices of a degree-$d$ form evolve under a linear change of
variables. This result is an extension of Theorem \ref{thm:P3structural} to the higher degree case.
\begin{theorem}\label{thm:Pdslicestucture}
Let $g$ be a degree-$d$ form with slices $\{S_{i_1...i_{d-2}}\}_{i_1,...,i_{d-2} \in [n]}$ and let $f(x) = g(Ax)$. Then the slices $T_{i_1...i_{d-2}}$ of $f$, are given by
\begin{align*}
    T_{i_1...i_{d-2}} = A^TD_{i_1...i_{d-2}}A
\end{align*}
where $D_{i_1...i_{d-2}} = \sum_{j_1...j_{d-2} \in [n]}a_{j_1i_1}...a_{j_{d-2}i_{d-2}}S_{j_1...j_{d-2}}$ and $a_{i,j}$ are the entries of $A$.
\newline
If $g = \sum_{i=1}^n \alpha_i x_i^d$, we have $D_{i_1...i_{d-2}} = \text{diag}(\alpha_1(\prod_{m=1}^{d-2}a_{1,i_m}),...,\alpha_n(\prod_{m=1}^{d-2}a_{n,i_m}))$.
\end{theorem}
\begin{proof}
We use the fact that
\begin{align*}
    S_{i_1...i_{d-2}} &= \frac{1}{d!}H_{\frac{\partial^{d-2} g}{\partial x_1....\partial x_{d-2}}}(x) \\
    T_{i_1...i_{d-2}} &= \frac{1}{d!}H_{\frac{\partial^{d-2} f}{\partial x_1....\partial x_{d-2}}}(x)
\end{align*}
where $H_f(x)$ is the Hessian matrix of $f$ at the point $x$.
\newline
Since $f(x) = g(Ax)$, by differentiating $d$ times, we get that
\begin{align*}
    \frac{\partial^{d} f}{\partial x_{i_1}...\partial x_{i_{d}}}(x) = \sum_{j_1...j_{d} \in [n]}a_{j_1i_1}...a_{j_{d}i_{d}}\frac{\partial^{d} g}{\partial x_{j_1}...\partial x_{j_{d}}}(Ax) 
\end{align*}
Putting these equations in matrix form, and using the fact that 
\begin{align*}
    \frac{\partial^d g}{\partial x_{j_1}...\partial x_{j_d}}(Ax) = \frac{\partial^d g}{\partial x_{j_1}...\partial x_{j_d}}(x)
\end{align*}
we get the following equation
\begin{align*}
    T_{i_1...i_{d-2}} = A^T(\sum_{j_1...j_{d-2} \in [n]}a_{j_1i_1}...a_{j_{d-2}i_{d-2}}S_{j_1...j_{d-2}})A.
\end{align*}
\end{proof}
The next lemma uses Theorem \ref{thm:Pdslicestucture} to reveal some crucial properties about the subspace spanned by the slices of any degree-$d$ form which is equivalent to some $g \in \mathcal{P}_d$. It is an extension of Lemma \ref{lem:subspace} to the higher degree case.
\begin{lemma}\label{lem:Pdstructural}
Let $f(x_1,...,x_n)$ and $g(x_1,...,x_n)$ be two forms of degree $d$ such that $f(x) = g(Ax)$ for some non-singular matrix $A$.
\begin{enumerate}
    \item If $\mathcal{U}$ and $\mathcal{V}$ denote the subspaces of $M_n(\mathbb{K})$ spanned respectively by the slices of $f$ and $g$, we have $\mathcal{U} = A^T\mathcal{V}A$.
    \item $\mathcal{V}$ is non-singular iff $\mathcal{U}$ is non-singular.
    \item In particular, for $g \in \mathcal{P}_d$ the subspace $\mathcal{V}$ is the space of diagonal matrices and $\mathcal{U}$ is a non-singular subspace, i.e., it is not made of singular matrices only.
\end{enumerate}
\end{lemma}
\begin{proof}
Theorem \ref{thm:Pdslicestucture} shows that $\mathcal{U} \subseteq A^T\mathcal{V}A$. Now since, $g(x) = f(A^{-1}x)$, same argument shows that $\mathcal{V} \subseteq A^{-T}\mathcal{U}A^{-1}$. This gives us that $\mathcal{U} = A^T\mathcal{V}A$.
\par
For the second part of the lemma, first let us assume that $\mathcal{V}$ is non-singular. Let us take any arbitrary $M_{\mathcal{U}} \in \mathcal{U}$. Using the previous part of the lemma, we know that there exists $M_{\mathcal{V}} \in \mathcal{V}$ such that $M_{\mathcal{U}} = A^TM_{\mathcal{V}}A$. Since $\mathcal{V}$ is non-singular, $\text{det}(M_{\mathcal{V}}) \neq 0$. Taking determinant on both sides, we get that $\text{det}(M_{\mathcal{U}}) = \text{det}(A)^2\text{det}(M_{\mathcal{V}}) \neq 0$ (since $A$ is invertible, $\text{det}(A) \neq 0$). 
\newline
For the converse, assume that $\mathcal{U}$ is non-singular. Let us take any arbitrary $M_{\mathcal{V}} \in \mathcal{V}$. Using the previous part of the lemma, we know that there exists $M_{\mathcal{U}} \in \mathcal{U}$ such that $M_{\mathcal{U}} = A^{-T}M_{\mathcal{V}}A^{-1}$. Since $\mathcal{U}$ is non-singular, $\text{det}(M_{\mathcal{U}}) \neq 0$. Taking determinant on both sides, we get that $\text{det}(M_{\mathcal{U}}) = \text{det}(A)^{-2}\text{det}(M_{\mathcal{V}}) \neq 0$ (since $A$ is invertible, $\text{det}(A) \neq 0$). 
\par
For the third part of the lemma, let $\{S_{i_1...i_{d-2}}\}_{i_1...i_{d-2} \in [n]}$ be the slices of $g$. If $g = \sum_{i \in [n]} \alpha_i x_i^d$, such that $\alpha_i \neq 0$ for all $i$, $S_{\bar{i}}$ has $\alpha_i$ in the $(i,i)$-th position and $0$ everywhere. Also, $S_{i_1,...,i_{d-2}} = 0$, when the $i_k$'s are not equal. Hence, $\mathcal{V}$ is the space of all diagonal matrices. Hence $\mathcal{V}$ is a non-singular space. Using the previous part of the lemma, we get that $\mathcal{U}$ is a non-singular space as well.

\end{proof}
The next lemma is effectively a converse of the second part of Lemma~ \ref{lem:Pdstructural}. It shows that if the slices of $f$ are diagonal matrices, then the fact that they effectively originate from a symmetric tensor force them to be extremely special.
\begin{lemma}\label{lem:charconverse}
Let $f \in \mathbb{K}[x_1,...,x_n]_d$ be a degree-$d$ form. If the slices of $f$ are diagonal matrices, then $f = \sum_{i \in [n]} \alpha_i x_i^d$ for some $\alpha_1,...,\alpha_n \in \mathbb{K}$.
\end{lemma}
\begin{proof}
Let $T_{i_1,...,i_{d-2}}$ be the slices of $f$. Let $I = \{(i_{\sigma(1)},...,i_{\sigma(d)})| \sigma \in S_d\}$. Now since they are slices of a polynomial, we know that
\begin{equation}\label{eq:slicediag}
     (T_{i_1...i_{d-2}})_{i_{d-1},i_d} = (T_{i_{\sigma(1)},...,i_{\sigma(d-2)}})_{i_{\sigma(d-1)},i_{\sigma(d)}}.
\end{equation}
We want to show that $T_{i_1,...,i_d} \neq 0$ only if $i_1 = i_2 = ... = i_{d}$. Using (\ref{eq:slicediag}), it is sufficient to show that $(T_{i_1,...,i_{d-2}})_{i_{d-1},i_d} \neq 0$ only if $i_{d-1} = i_d$. This is true since $T_{i_1,...,i_{d-2}}$ are diagonal matrices.
\newline
This gives us that $f = \sum_{i \in [n]} \alpha_ix_i^d$.
\end{proof}
Now we are finally ready to prove a theorem that characterizes exactly the set of degree-$d$ homogeneous polynomials which are equivalent to some $g \in \mathcal{P}_d$. This is an extension of Theorem \ref{thm:socchar} to the degree-$d$ case, and it already appears as Theorem~\ref{thm:Pdcharintro} in the introduction. We restate it now for the reader's convenience.
\begin{theorem}\label{thm:Pdchar}
A degree $d$ form $f \in \mathbb{K}[x_1,...,x_n]$ is equivalent to some polynomial $P_d \in \mathcal{P}_d$ if and only if its slices $\{T_{i_1,...,i_{d-2}}\}_{i_1,...,i_{d-2} \in [n]}$ span a non-singular matrix space and the slices are simultaneously diagonalisable by congruence, i.e., there exists an invertible matrix $Q \in M_n(\mathbb{K})$ such that the matrices $Q^T T_{i_1...i_{d-2}} Q$ are diagonal for all $i_1,...,i_{d-2} \in [n]$.
\end{theorem}
\begin{proof}
Let $\mathcal{U}$ be the space spanned by $\{T_{i_1,...,i_{d-2}}\}_{i_1,...,i_{d-2} \in [n]}$ . If $f$ is equivalent to $P_d$, Theorem \ref{thm:Pdslicestucture} shows that the slices of $f$ are simultaneously diagonalisable by congruence and Lemma \ref{lem:Pdstructural}  shows that $\mathcal{U}$ is non-singular.
\par
Let us show the converse. Since the slices $\{T_{i_1,...,i_{d-2}}\}_{i_1,...,i_{d-2} \in [n]}$  are simultaneously diagonalisable, there are diagonal matrices $\Lambda_{i_1...i_{d-2}}$ and a non-singular matrix $R \in M_n(\mathbb{K})$ such that
\begin{align*}
    T_{i_1...i_{d-2}} = R\Lambda_{i_1...i_{d-2}}R^T \text{ for all } i_1,...,i_{d-2} \in [n].
\end{align*}
So now we consider $g(x) = f(R^{-T}x)$. Let $\{S_{i_1,...,i_{d-2}}\}_{i_1,...,i_{d-2} \in [n]}$ be the slices of $g$. Using Theorem (\ref{thm:Pdslicestucture}), we get that 
\begin{align*}
    S_{i_1...i_{d-2}} &= (R^{-T})(\sum_{j_1...j_{d-2} \in [n]}r_{j_1i_1}...r_{j_{d-2}i_{d-2}}R\Lambda_{j_1...j_{d-2}}R^T)R^{-T} \\
    &= \sum_{j_1...j_{d-2} \in [n]}r_{j_1i_1}...r_{j_{d-2}i_{d-2}}\Lambda_{j_1...j_{d-2}}.
\end{align*}
This implies that $S_{j_1...j_{d-2}}$ are also diagonal matrices. 
By Lemma \ref{lem:charconverse}, $g = \sum_{i \in [n]} \alpha_i x_i^{d}$. It therefore remains to be shown that $\alpha_i \neq 0$, for all $i \in [n]$.
Let $\mathcal{V}$ be the subspace spanned by the slices of $g$ and the slices of $f$ span a non-singular matrix space $\mathcal{U}$.
Since,  $\mathcal{U}$ is a non-singular subspace of matrices, using part (2) of Lemma \ref{lem:Pdstructural}, we get that $\mathcal{V}$ is a non-singular subspace of matrices.
%\newline

But if some $\alpha_i$ vanishes, for all $A \in \mathcal{V}$, $A_{\bar{i}} = 0$. Hence $\mathcal{V}$ is a singular subspace, which is a contradiction.
%\newline
This gives us that $g = \sum_{i=1}^n \alpha_i x_i^d$ where $\alpha_i \neq 0$ for all $i$. Hence, $g \in \mathcal{P}_d$ and $f$ is equivalent to $g$.
\end{proof}
\begin{theorem}\label{thm:sodchar}
Let $f \in \mathbb{K}[x_1,...,x_n]$ be a degree-$d$ form. $f$ is equivalent to some polynomial $P_d \in \mathcal{P}_d$ iff the subspace $\mathcal{V}$ spanned by its slices $\{T_{i_1,...,i_{d-2}}\}_{i_1,...,i_{d-2} \in [n]}$ is a non-singular subspace and $\mathcal{V}$ satisfies the Commutativity Property and the Diagonalisability Property.
\end{theorem}
\begin{proof}
This follows from Theorem \ref{thm:Pdchar} and Theorem \ref{thm:main} for $k = n^{d-2}$ to get the result. 
\end{proof}

We now state a weaker definition of the singularity of a subspace spanned by a set of matrices and using that we prove a stronger version of Theorem~\ref{thm:sodchar}. More formally we show that the characterization is valid even when the "non-singular subspace" criterion imposed on the subspace $\mathcal{V}$ spanned by the slices of the polynomial is replaced by the "not a weakly singular subspace" criterion.
\begin{definition} (Weak singularity)
\newline
Let $\mathcal{V}$ be the space spanned by matrices $\{S_{i_1,...,i_{d-2}}\}_{i_1,...,i_{d-2} \in [n]}$ . $\mathcal{V}$ is weakly singular if for all $\alpha = (\alpha_1,...,\alpha_n) $, 
\begin{align*}
    \text{det}(\sum_{i_1,...,i_{d-2}\in [n]} (\prod_{k \in [d-2]}\alpha_{i_k}) S_{i_1...i_{d-2}}) = 0.
\end{align*}
\end{definition}
Notice here that the notion of weak-singularity is entirely dependent on the generating set of matrices. So it is more of a property of the generating set. But by abuse of language, we will call the span of the matrices to be weakly singular. To put it in contrast, refer to Section \ref{sec:cubes} where the notion of singularity is a property of the subspace spanned by the matrices (irrespective of the generating set).
\begin{theorem}\label{sodcharspl}
Let $f \in \mathbb{K}[x_1,...,x_n]$ be a degree-$d$ form. $f$ is equivalent to some polynomial  $P_d \in \mathcal{P}_d$ iff the subspace $\mathcal{V}$ spanned by its slices $\{T_{i_1,...,i_{d-2}}\}_{i_1,...,i_{d-2} \in [n]}$ is not a weakly singular subspace, satisfies the Commutativity Property and the Diagonalisability Property.
\end{theorem}
\begin{proof}
First we show that if $f = P_d(Ax)$ such that $P_d \in \mathcal{P}_d$ i.e. $P_d(x) = \sum_{i=1}^n \alpha_ix_i^d$ where  $\alpha_i \neq 0$ for all $i \in [n]$ and $A$ is invertible, then $\mathcal{V}$ is not a weakly singular subspace, satisfies the commutativity property and the diagonalisability property.
\newline
Let $\{S_{i_1...i_{d-2}}\}_{i_1,...,i_{d-2} \in [n]}$ be the slices of $P_d$. Then $S_{\bar{i}} = \alpha_i \text{diag}(e_i)$ where $e_i$ is the $i$-th standard basis vector, and all other slices are $0$. From Theorem~\ref{thm:Pdslicestucture}
\begin{align*}
    T_{i_1...i_{d-2}} &= A^TD_{i_1...i_{d-2}}A \\
    &= A^T(\sum_{k \in [n]}a_{ki_1}...a_{ki_{d-2}}S_{\bar{k}})A.
\end{align*}
    Now we define 
    \begin{align*}
        T(\bar{\beta}) &= \sum_{i_1,...,i_{d-2} \in [n]} (\prod_{k \in [d-2]}\beta_{i_k}) T_{i_1...i_{d-2}} \\
        &= \sum_{i_1,...,i_{d-2} \in [n]} (\prod_{k \in [d-2]}\beta_{i_k})A^T(\text{diag}(\alpha_1(\prod_{m \in [d-2]}a_{1i_m}),...,\alpha_n(\prod_{m \in [d-2]}a_{ni_m})))A \\
        &= A^T\text{diag}(\alpha_1(\sum_{i_1,...,i_{d-2} \in [n]}(\prod_{k \in [d-2]}\beta_{i_k}a_{1i_k})),...,\alpha_n(\sum_{i_1,...,i_{d-2} \in [n]}(\prod_{k \in [d-2]}\beta_{i_k}a_{ni_k})))A.
    \end{align*}
Taking determinant on both sides,
\begin{align*}
    \text{det}(T)(\bar{\beta}) = \text{det}(A)^2\prod_{m=1}^nT_m(\bar{\beta})
\end{align*}
where $T_m(\bar{\beta}) = \alpha_m(\sum_{i_1,...,i_{d-2} \in [n]}(\prod_{k \in [d-2]}\beta_{i_k}a_{mi_k}))$.
\newline
Since, $A$ is invertible, none of its rows are all $0$. Hence for all $m_0 \in [n]$, there exists $j_0 \in [n]$, such that $a_{m_0j_0} \neq 0$. Then
\begin{align*}
    \text{coeff}_{\beta_{j_0}^{d-2}}(T_{m_0}) = a_{m_0j_0}^{d-2} \neq 0.
\end{align*}
Hence $T_{m_0} \not\equiv 0$ for all $m_0 \in [n]$ which implies that $\text{det}(T) \not\equiv 0$. Therefore, there exists $ \bar{\beta^0}$ such that $\text{det}(T)(\bar{\beta^0}) \neq 0$.
\newline
This proves that 
\begin{align*}
    \text{det}(\sum_{i_1,...,i_{d-2} \in [n]} (\prod_{k \in [d-2]}\beta_{i_k}) T_{i_1...i_{d-2}}) \not\equiv 0.
\end{align*}
Hence, $\mathcal{V} = \text{span}\{T_{i_1...i_{d-2}}\}_{i_1,...,i_{d-2} \in [n]}$ is not weakly singular. Theorem \ref{thm:sodchar} gives us that the subspace spanned by the slices $\mathcal{V}$ satisfies the commutativity property and the diagonalisability property.
\par
For the converse, if $\mathcal{V}$ is not a weakly singular subspace, then it is a non-singular subspace as well. And it satisfies the commutativity property and the diagonalisability property. By Theorem \ref{thm:sodchar}, we get that $f$ is equivalent to some polynomial in $\mathcal{P}_d$.
\end{proof}

\subsection{Analysis for positive inputs}\label{sec:positived}
In this section we analyze the behavior of Algorithm \ref{alg:pd} on inputs that are
equivalent to some polynomial in $\mathcal{P}_d$. We recall here again that by $T_{\bar{1}}$, we denote the slice $T_{11...1}$.
\begin{lemma}\label{lem:pdt1invertible}
Let $f \in \mathbb{K}[x_1,...,x_n]_d$ with slices $\{S_{i_1,...,i_{d-2}}\}_{i_1,...,i_{d-2} \in [n]}$, such that the subspace $\mathcal{V}$ spanned by the slices is not weakly singular.
Let $h(x) = f(Rx)$ where the entries $r_{i,j}$ are chosen uniformly and independently at random from a finite set $S \subseteq \mathbb{K}$. Let $\{T_{i_1,...,i_{d-2}}\}_{i_1,...,i_{d-2} \in [n]}$ be the slices of $h$. Then
\begin{align*}
    \text{Pr}_{R \in S}[T_{\bar{1}} \text{ is invertible}] \geq 1- \frac{n(d-1)}{|S|}.
\end{align*}
\end{lemma}
\begin{proof}
 We can obtain the slices $T_{i_1...i_{d-2}}$ of $h$ from the slices $S_{i_1...i_{d-2}}$ of $f$ using Theorem \ref{thm:Pdslicestucture}. Namely, we have 
 \begin{align*}
     T_{i_1...i_{d-2}} = R^TD_{i_1...i_{d-2}}R
 \end{align*}
 where 
 \begin{align*}
     D_{i_1...i_{d-2}} = \sum_{j_1...j_{d-2} \in [n]}( \prod_{m \in [d-2]}r_{j_m,i_m}) S_{j_1...j_{d-2}}.
 \end{align*}
 Therefore $T_{\bar{1}}$ is invertible iff $R$ and $D_{\bar{1}}$ are invertible.
 Applying Schwartz-Zippel lemma to $\text{det}(R)$ shows that $R$ is singular with probability at most $\frac{n}{|S|}$. We will show that $D_{\bar{1}}$ is singular with probability at most $\frac{n(d-2)}{|S|}$. The lemma then follows from the union bound. Matrix $D_{\bar{1}}$ is not invertible iff $\text{det}(D_{\bar{1}}) = 0$. Since, $D_{\bar{1}} = \sum_{j_1...j_{d-2} \in [n]}( \prod_{m \in [d-2]}r_{j_m,1}) S_{j_1...j_{d-2}}$, $\text{det}(D_{\bar{1}}) \in \mathbb{K}[r_{1,1},...,r_{n,1}]$ and $\text{deg}(\text{det}(D_{\bar{1}})) \leq n(d-2)$. Since, $\mathcal{V}$ is not weakly singular, there exists some choice of $\alpha = (\alpha_{1},...,\alpha_{n})$, such that
 \begin{align*}
      S = \sum_{i_1,...,i_{d-2} \in [n]}(\prod_{m \in [d-2]}\alpha_{i_m}) S_{i_1...i_{d-2}}
 \end{align*}
 is invertible. Hence, $\text{det}(S) \neq 0$. This gives us that $\text{det}(D_{\bar{1}})(\alpha) \neq 0$. which gives us that $\text{det}(D_{\bar{1}}) \not\equiv 0 $. From the Schwartz-Zippel lemma, it follows  that
 \begin{align*}
     \text{Pr}_{R \in S}[\text{det}(D_{\bar{1}}) = 0] \leq \frac{n(d-2)}{|S|}.
 \end{align*}
\end{proof}
Recall here from Section \ref{sec:notations}, we define by $\mathcal{P}_d$, the set of all polynomials of the form $\sum_{i=1}^n \alpha_ix_i^d$ such that $0 \neq \alpha_i \in \mathbb{K}$ for all $i \in [n]$.
\begin{lemma}\label{lem:pdopp}
Given $A \in M_n(\mathbb{K})$, let $\{T_{i_1,...,i_{d-2}}\}_{i_1,...,i_{d-2} \in [n]}$ be the slices of $h(x) = P_d(Ax)$ where $P_d\in \mathcal{P}_d$. If $T_{\bar{1}}$ is invertible, define $T'_{\bar{1}} = (T_{\bar{1}})^{-1}$. Then
 $T'_{\bar{1}}T_{\bar{2}}$ commutes with $T'_{\bar{1}}T_{\bar{3}}$ and $T'_{\bar{1}}T_{\bar{2}}$ is diagonalisable.
\end{lemma}
\begin{proof}
Let $P_d = \sum_{i=1}^n \alpha_i x_i^d$ where $\alpha_i \neq 0$. By Theorem \ref{thm:Pdslicestucture}, 
$$T_{i_1...i_{d-2}} = A^T(\text{diag}(\alpha_1(\prod_{m=1}^{d-2}a_{1,i_m}),...,\alpha_n(\prod_{m=1}^{d-2}a_{n,i_m})))A = A^TD_{i_1...i_{d-2}}A.$$
%\begin{align*}
%    T_{\bar{1}} &= R^T(\text{diag}(r_1,...,r_{n1}))R = R^TD_{\bar{1}}R \\
 %   T_{\bar{2}} &= R^T(\text{diag}(r_{12},...,r_{n2}))R = R^TD_2R \\
%    T_{\bar{3}} &= R^T(\text{diag}(r_{13},...,r_{n3}))R = R^TD_3R
%\end{align*}
If $T_{\bar{1}}$ is invertible, the same is true of $A$ and $D_{\bar{1}}$.
The inverse $(D_{\bar{1}})^{-1}$ is diagonal like $D_{\bar{1}}$, hence $(D_{\bar{1}})^{-1}D_{\bar{2}}$ and $(D_{\bar{1}})^{-1}D_{\bar{3}}$ are both diagonal as well and must therefore commute. Now,
\begin{align*}
    T'_{\bar{1}}T_{\bar{2}}T'_{\bar{1}}T_{\bar{3}} &= A^{-1}((D_{\bar{1}})^{-1}D_{\bar{2}}(D_{\bar{1}})^{-1}D_{\bar{3}})A \\
    &= A^{-1}((D_{\bar{1}})^{-1}D_{\bar{3}}(D_{\bar{1}})^{-1}D_{\bar{2}})A \\
    &= T'_{\bar{1}}T_{\bar{3}}T'_{\bar{1}}T_{\bar{2}}.
\end{align*}
\par
Finally,  $T'_{\bar{1}}T_{\bar{2}} = A^{-1}((D_{\bar{1}})^{-1}D_{\bar{2}})A$ so this matrix is diagonalisable.
\end{proof}
We are now in a position to prove the easier half of Theorem \ref{thm:pdproof}.
\begin{theorem}
If an input $f \in \mathbb{K}[x_1,...,x_n]_d$ is equivalent to some polynomial $P_d \in \mathcal{P}_d$ then $f$ will be accepted by Algorithm \ref{alg:pd} with high probability over the choice of the random matrix $R$. More precisely, if the entries $r_{i,j}$ are chosen uniformly and independently at random from a finite set $S \subseteq \mathbb{K}$, then the input will be accepted with probability $\geq (1-\frac{n(d-1)}{|S|})$.
\end{theorem}
\begin{proof}
We start by assuming that $f  = P_d(Bx)$ for some $P_d \in \mathcal{P}_d$ where $B$ is an invertible matrix. 
%By Lemma \ref{lem:Pdstructural}, 
%We can obtain the slices $T_{i_1...i_{d-2}}$ of $h$ from the slices $S_{i,j}$ of $f$ using Theorem \ref{thm:Pdslicestucture}. Namely, we have $T_{i_1...i_{d-2}} = R^TD_{i_1...i_{d-2}}R$ where $D_{i_1...i_{d-2}} = \sum_{j_1...j_{d-2} \in [n]} r_{j_1,i_1}...r_{j_{d-2},i_{d-2}}S_{j_1...j_{d-2}}$.
By Theorem \ref{sodcharspl}, we know that the subspace spanned by the slices of $f$ is not weakly singular. We can therefore apply Lemma~\ref{lem:pdt1invertible}, the first slice $T_{\bar{1}}$ of $h(x) = f(Rx)$  is invertible with probability at least $1- \frac{n(d-1)}{|S|}$. Moreover if $T_{\bar{1}}$ is invertible, Lemma \ref{lem:pdopp} shows that, $f$ will always be accepted. (We can apply this lemma to $h$ since $h = P_d(RBx)$).
\end{proof}

\subsection{Failure of commutativity} \label{sec:failcomd}
In this section we first give a suitable generalization of Lemma \ref{lem:comm} required for the degree-$d$ case. This is essential for the correctness proof of Algorithm~\ref{alg:pd}. 
\begin{definition}\label{def:symmfamily}
Let $\{S_{i_1,...,i_d}\}_{i_1,...,i_d \in [n]}$ be a family of matrices. We say that the matrices form a symmetric family of symmetric matrices if each matrix in the family is symmetric and for all permutations $\sigma \in S_d$, $S_{i_1,...,i_d} = S_{i_{\sigma(1)}...i_{\sigma(d)}}$.
\end{definition}

\begin{lemma}[General commutativity lemma]\label{lem:gencomm}
Let $\{S_{i_1,...,i_d}\}_{i_1,...,i_d \in [n]}$ be a symmetric family of symmetric matrices in $M_n(\mathbb{K})$ such that they do not form a commuting family. Pick $\alpha = \{\alpha_1,...,\alpha_n\}$ and $\alpha' = \{\alpha'_1,...,\alpha'_n\}$  uniformly and independently at random from a finite set $S \subset \mathbb{K}$. We define 
\begin{align*}
    M_{\alpha} &= \sum_{i_1,...,i_d \in [n]} (\prod_{m \in [d]}\alpha_{i_m}) S_{i_1,...,i_d} \\
    M_{\alpha'} &= \sum_{j_1,...,j_d \in [n]} (\prod_{m \in [d]}\alpha'_{j_m}) S_{j_1,...,j_d}.
\end{align*}
Then,
\begin{align*}
    \text{Pr}_{\alpha,\alpha' \in S}\Big[ M_{\alpha}, M_{\alpha'} \text{ don't commute}\Big] \geq \Big(1-\frac{2d}{|S|}\Big).
\end{align*}
\end{lemma}
\begin{proof}
We want to bound the probability of error, i.e
\begin{align*}
    \text{Pr}_{\alpha,\alpha' \in S}\Big[ M_{\alpha} M_{\alpha'} -  M_{\alpha'} M_{\alpha} \neq 0\Big].
\end{align*}
Now 
\begin{align*}
    & M_{\alpha} M_{\alpha'} -  M_{\alpha'} M_{\alpha} \\
    = &\sum_{i_1,...,i_d,j_1,...,j_d \in [n]} (\prod_{m \in [d]}\alpha_{i_m}\alpha'_{j_m}) (S_{i_1...i_d}S_{j_1...j_d} - S_{j_1...j_d}S_{i_1...i_d}). 
\end{align*}
For a fixed $r,s \in [n]$, we define the polynomial
\begin{align*}
    P^{r,s}_{\text{comm}}(\alpha,\alpha') = \sum_{i_1,...,i_d,j_1,...,j_d \in [n]} (\prod_{m \in [d]}\alpha_{i_m}\alpha'_{j_m})m^{r,s}_{i_1...i_dj_1...j_d}
\end{align*}
where 
\begin{align*}
     m^{r,s}_{i_1...i_dj_1...j_d} = (S_{i_1...i_d}S_{j_1...j_d} - S_{j_1...j_d}S_{i_1...i_d})_{r,s}.
\end{align*}
First note that by construction $ M_{\alpha}$ commutes with $ M_{\alpha'}$ if and only if for all $r,s \in [n]$ such that $P^{r,s}_{\text{comm}}(\alpha,\alpha') = 0$.
\newline
Since, $\{S_{i_1,...,i_d}\}$ is not a commuting family, there exists $i_1^0,...,i_d^0,j_1^0,...,j_d^0 \in [n]$, such that 
\begin{equation} \label{eq:fixedcomm}
    S_{i_1^0...i_d^0}S_{j_1^0...j_d^0} -  S_{j_1^0...j_d^0}S_{i_1^0...i_d^0} \neq 0.
\end{equation}
Hence, there exists some entry $(r_0,s_0)$ such that 
\begin{equation}
    (S_{i_1^0...i_d^0}S_{j_1^0...j_d^0} -  S_{j_1^0...j_d^0}S_{i_1^0...i_d^0})_{r_0,s_0} \neq 0.
\end{equation}
Now we claim that $P^{r_0,s_0}_{\text{comm}}(\alpha,\alpha') \not\equiv 0 $. It is enough to show that the coefficient of $\alpha_{i_1^0}...\alpha_{i_d^0}\alpha'_{j_1^0}...\alpha'_{j_d^0}$ in $P^{r_0,s_0}_{\text{comm}}(\alpha,\alpha')$ is non-zero. Let $I_0 = \{(i^0_{\sigma(1)},...,i^0_{\sigma(d)})|\sigma \in~S_d\}$ and let $J_0 = \{(j^0_{\sigma(1)},...,j^0_{\sigma(d)})|\sigma \in S_d\}$. Then
%Let $P_{I_0}$ be the set of permutations of $I_0 = (i_1^0,...,i_d^0)$ and let $P_{J_0}$ be the set of permutations of $J_0 = (j_1^0,...,j_d^0)$
\begin{align*}
    \text{coeff}_{\alpha_{i_1^0}...\alpha_{i_d^0}\alpha'_{j_1^0}...\alpha'_{j_d^0}}(P^{r_0,s_0}_{\text{comm}}) = \sum_{\bar{i} \in I_0,\bar{j} \in J_0} m^{r_0s_0}_{\bar{i}\bar{j}}.
\end{align*}
The matrices $S_{i_1...i_d}$ form a symmetric family in the sense of Definition \ref{def:symmfamily}. Therefore, for all $\bar{i} \in I_0,\bar{j} \in J_0$, $m^{r_0s_0}_{\bar{i}\bar{j}}$ are equal.
This gives us that
\begin{align*}
        \text{coeff}_{\alpha_{i_1^0}...\alpha_{i_d^0}\alpha'_{j_1^0}...\alpha'_{j_d^0}}(P^{r_0,s_0}_{\text{comm}}) = |I_0||J_0|(m^{r_0,s_0}_{i_1^0...i_d^0j_1^0...j_d^0}) \neq 0.
\end{align*}
Hence $P^{r_0,s_0}_{\text{comm}} \not \equiv 0$ and $\text{deg}(P^{r_0,s_0}_{\text{comm}}) \leq 2d$ and using Schwartz-Zippel lemma, we get that,
\begin{align*}
    \text{Pr}_{\alpha,\alpha' \in S}[P^{r_0,s_0}_{\text{comm}}(\alpha,\alpha') \neq 0] \geq 1- \frac{2d}{|S|}. 
\end{align*}
Putting $r=r_0,s=s_0$, this gives us that
\begin{align*}
    \text{Pr}_{\alpha,\alpha' \in S}\Big[ M_{\alpha}, M_{\alpha'} \text{ don't commute}\Big] \geq \Big(1-\frac{2d}{|S|}\Big).
\end{align*}
\end{proof}
Note that if in a certain index, $i^0_1$ occurs $n_1$ times,$i^0_2$ occurs $n_2$ times,...,$i^0_r$ occurs $n_r$ times, then $|I_0| = \frac{d!}{n_1!...n_r!}$.
\par
The next result relies on the above lemma. Theorem \ref{thm:pdalgcomm} gives us a way
to analyze the case when the slices of the input polynomial fail to satisfy
the commutativity property (recall that this property is relevant due to
Theorem~\ref{sodcharspl}). This can also be seen as a suitable extension of Theorem~\ref{thm:algcomm} to the general degree-$d$ case.
\begin{theorem}\label{thm:pdalgcomm}
Let $f \in \mathbb{K}[x_1,...,x_n]_d$ be a degree $d$ form such that the subspace of matrices $\mathcal{V}$ spanned by its slices is not weakly singular and does not satisfy the commutativity property.
Let $h(x) = f(Rx)$ where the entries $(r_{i,j})$ of $R$ are chosen uniformly and independently at random from a finite set $S \subset \mathbb{K}$. Let $\{T_{i_1...i_{d-2}}\}_{i_1,...,i_{d-2} \in [n]}$ be the slices of $h$. 
If $T_{\bar{1}}$ is invertible, define $T'_{\bar{1}} = (T_{\bar{1}})^{-1}$. Then 
\begin{align*}
    \text{Pr}[T_{\bar{1}} \text{ is invertible and } T'_{\bar{1}}T_{\bar{2}}, T'_{\bar{1}}T_{\bar{3}} \text{ commute } ] \leq \frac{2(d-2)}{|S|}.
\end{align*}
\end{theorem}
\begin{proof}
Let $\{S_{i_1...i_{d-2}}\}_{i_1,...,i_{d-2} \in [n]}$ be the slices of $f$. By Theorem \ref{thm:Pdslicestucture}, we know that $$T_{i_1...i_{d-2}} = R^T(\sum_{j_1...j_{d-2} \in [n]} (\prod_{m \in [d-2]}r_{j_m,i_m}S_{j_1...j_{d-2}})R.$$ Let us define $D_{i_1...i_{d-2}} = \sum_{j_1...j_{d-2} \in [n]} (\prod_{m \in [d-2]}r_{j_m,i_m}) S_{j_1...j_{d-2}}$.
Then we have:
\begin{align}
    T'_{\bar{1}}T_{\bar{2}} &= R^{-1}(D_{\bar{1}})^{-1}(R)^{-T}R^TD_{\bar{2}}R \nonumber \\
    &= R^{-1}(D_{\bar{1}})^{-1}D_{\bar{2}}R  \nonumber \\
    &=R^{-1}(D_{\bar{1}})^{-1}(\sum_{j_1...j_{d-2} \in [n]} (\prod_{m \in [d-2]}r_{j_m,2})S_{j_1...j_{d-2}})R \nonumber \\ 
    & = R^{-1}(D_{\bar{1}})^{-1}(\sum_{j_1...j_{d-2} \in [n]} (\prod_{m \in [d-2]}r_{j_m,2})S_{j_1...j_{d-2}})R \nonumber \\ 
    &= R^{-1}(\sum_{j_1...j_{d-2} \in [n]} (\prod_{m \in [d-2]}r_{j_m,2})(D_{\bar{1}})^{-1}S_{j_1...j_{d-2}})R .\label{eq:comm}
\end{align}
Similarly 
\begin{align*}
    T'_{\bar{1}}T_{\bar{3}} &= R^{-1}(D_{\bar{1}})^{-1}D_{\bar{3}}R \\ &=R^{-1}(\sum_{j_1...j_{d-2} \in [n]} (\prod_{m \in [d-2]}r_{j_m,3})(D_{\bar{1}})^{-1}S_{j_1...j_{d-2}})R.
\end{align*}
So, $T'_{\bar{1}}T_{\bar{2}}$ commutes with $T'_{\bar{1}}T_{\bar{3}}$ iff $R$ is invertible and $(D_{\bar{1}})^{-1}D_{\bar{2}}$ commutes with $(D_{\bar{1}})^{-1}D_{\bar{3}}$.
\par
Let $E_1$ be the event such that $T_{\bar{1}}'T_{\bar{2}}$ commutes with $T_{\bar{1}}'T_{\bar{3}}$ and let $E_1'$ be the event such that $(D_{\bar{1}})^{-1}D_{\bar{2}}$ commutes with $(D_{\bar{1}})^{-1}D_{\bar{3}}$.
Let $E_2$ be the event such that $\{(D_{\bar{1}})^{-1}S_{i_1,...,i_{d-2}}\}_{i_1,...,i_{d-2} \in [n]}$ is not a commuting family. Since $\mathcal{V}$ does not satisfy the commutativity property, $(D_{\bar{1}})^{-1}\mathcal{V}$ is not a commuting subspace if $D_{\bar{1}}$ is invertible.
Hence, the event such that $D_{\bar{1}}$ is invertible is the same as $E_2$.
Setting 
\begin{align*}
    A_{i_1...i_{d-2}} = (D_{\bar{1}})^{-1}S_{i_1...i_{d-2}}, \alpha_{i} = r_{i,2}, \alpha'_i = r_{i,3}
\end{align*}
and then using Lemma \ref{lem:gencomm} for $d= d-2$, we can conclude that 
\begin{equation}
    \label{eqn:pdcommbig}
    \text{Pr}_{R \in S}\Big[E_1'|E_2 \Big] \leq \frac{2(d-2)}{|S|}.
\end{equation}
Note here that $D_{\bar{1}}$ depends only on the random variables $r_{i,1}$ for all $i \in [n]$ and therefore is independent of $r_{k,2}$ and $r_{l,3}$ for all $k,l \in [n]$, because we assume that the entries of $R$ are all picked uniformly and independently at random.
\newline
Now we know that $T_{\bar{1}}$ is invertible iff $R$ and $D_{\bar{1}}$ are invertible.
Let $E_3$  be the event that $T_{\bar{1}}$ is invertible and $E_4$ the event that $R$  is invertible. We have $E_3 = E_2 \cap E_4$ and we have seen that $E_1 = E_1' \cap E_4$
Hence, the probability of error can be bounded as follows:
\begin{align*}
     \text{Pr}_{R \in S}[E_1 \cap E_3] = \text{Pr}_{R \in S}[E_1'\cap E_2 \cap E_4] 
     \leq \text{Pr}_{R \in S}[E_1'| E_2] \leq \frac{2(d-2)}{|S|} .
 \end{align*}
\end{proof}

\subsection{Failure of diagonalisability}\label{sec:faildiagd}
Theorem \ref{thm:pdalgcomm} gives us a way to analyze the case when the slices of the input polynomial fail to satisfy the commutativity property. With the results in
the present section we will be able to analyze the case where the commutativity property is satisfied, but the diagonalisability property fails (recall
that these properties are relevant due to Theorem \ref{sodcharspl}). This can also be seen as a suitable extension of Theorem \ref{thm:diagonalisability} to the general degree-$d$ case.
\begin{lemma}\label{lem:gendiag}
Let $\{A_{i_1...i_d}\}_{i_1,..,i_d \in [n]}  \in M_n(\mathbb{K})$ be a commuting family of matrices. Let us assume that this family is symmetric in the sense of Definition~\ref{def:symmfamily} and there exists $i_1^0,...,i_d^0 \in [n]$ such that $A_{i_1^0...i_d^0}$ is not diagonalisable. Let $S \subset \mathbb{K}$ be a finite set. Then  $D = \sum_{i_1,...,i_d =1}^n (\prod_{m \in [d]}\alpha_{i_m}) A_{i_1...i_d}$ is diagonalisable with probability at most $ \frac{d}{|S|}$ when $\alpha_1,...,\alpha_n$ are chosen uniformly and independently at random from $S$.
\end{lemma}
\begin{proof}
We define $\mathcal{U} = \text{span}\{A_{i_1...i_d}\}_{i_1,..,i_d \in [n]}$. 
We define 
\begin{align*}
    \mathcal{M} := \Big\{M \Big| M \text{ is diagonalisable and } M \in \mathcal{U}\Big\}.
\end{align*}
So the probability of error is 
\begin{equation}
\label{eq:diagprob}
    \text{Pr}_{\alpha \in S}\Big[D \in \mathcal{M}\Big].
\end{equation}
Now using Proposition \ref{prop:diagon}, and the hypothesis that there exists $A_{i_1^0...i_d^0} \in \mathcal{U} \setminus \mathcal{M}$, we get that $\mathcal{M}$ is a proper linear subspace of $\mathcal{U}$. So $\mathcal{M}$ is an intersection of hyperplanes.

Since $A_{i_1^0...i_d^0} \not\in \mathcal{M}$, there exists a linear form $l_{\mathcal{M}}(X) = \sum_{i,j \in [n]}a_{ij}x_{ij}$ corresponding to a hyperplane such that $l_{\mathcal{M}}(M) = 0$ for all $M \in \mathcal{M}$ and $l_{\mathcal{M}}(A_{i_1^0...i_d^0}) \neq 0$. We know that if $D$ is diagonalisable, then $l_{\mathcal{M}}(D) \neq 0$.

Now we want to compute 
\begin{align*}
    l_{\mathcal{M}}(D)(\alpha) &= \sum_{k,l \in [n]} a_{kl} \sum_{i_1,...,i_d\in [n]} (\prod_{m \in [d]}\alpha_{i_m}) (A_{i_1...i_d})_{k,l} \\
    &= \sum_{i_1,...,i_d\in [n]} (\prod_{m \in [d]}\alpha_{i_m}) m_{i_1...i_d}
\end{align*}
where
\begin{align*}
    m_{i_1...i_d} = (\sum_{k,l \in [n]} a_{kl}(A_{i_1...i_d})_{k,l}).
\end{align*}
Now we claim that $l_{\mathcal{M}}(D) \not\equiv 0$. We show this by proving that the coefficient of $\alpha_{i_1^0}...\alpha_{i_d^0}$ in $l_{\mathcal{M}}(D)(\alpha)$ is not equal to $0$. Let $I_0 = \{(i^0_{\sigma(1)},...,i^0_{\sigma(d)})|\sigma \in S_d\}$. Then
\begin{align*}
    \text{coeff}_{\alpha_{i_1^0}...\alpha_{i_d^0}}(D) = \sum_{\bar{i} \in I_0} m_{\bar{i}}.
\end{align*}
Since the matrices $A_{i_1...i_d}$ form a symmetric family, the $m_{\bar{i}}$ are equal for all $\bar{i} \in I_0$. Also, since, $l_{\mathcal{M}}(A_{i_1^0...i_d^0}) \neq 0$, we get that 
\begin{align*}
    \sum_{k,l \in [n]}a_{k,l}(A_{i_1^0...i_d^0})_{k,l} \neq 0.
\end{align*}
This gives us that $m_{i_1^0...i_d^0} \neq 0$. Hence, we get that
\begin{equation}
    \text{coeff}_{\alpha_{i_1^0}...\alpha_{i_d^0}}(D) = |I_0|m_{i_1^0...i_d^0} \neq 0.
\end{equation}
Thus, $l_{\mathcal{M}}(D) \not\equiv 0$ and $\text{deg}(l_{\mathcal{M}}(D)) \leq d$. Using Schwartz-Zippel lemma, the probability of error satisfies:
\begin{align*}
    \text{Pr}_{\alpha \in S}\Big[D \in \mathcal{M}\Big] \leq \text{Pr}_{\alpha \in S}[l_{\mathcal{M}}(D)(\alpha) = 0] \leq \frac{d}{|S|}.
\end{align*}
\end{proof}
Recall from Section \ref{sec:faildiagd} that if $i^0_1$ occurs $n_1$ times,$i^0_2$ occurs $n_2$ times,...,$i^0_r$ occurs $n_r$ times, then $|I_0| = \frac{d!}{n_1!...n_r!}$.
\par
The last result for this section is an analogue of the Theorem \ref{thm:pdalgcomm} for the diagonalisability property.
\begin{theorem}\label{thm:pdiagonalisability}
Let $f \in \mathbb{K}[x_1,...,x_n]_d$ be a degree-d form with such that the subspace $\mathcal{V}$ spanned by its slices  is a not weakly-singular subspace, satisfies the commutativity property, but does not satisfy the diagonalisability property. Let $h(x) = f(Rx)$ where the entries $r_{i,j}$ of $R$ are chosen uniformly and independently at random from a finite set $S \subset \mathbb{K}$, Let $\{T_{i_1...i_{d-2}}\}_{i_1,...,i_{d-2} \in [n]}$ be the slices of $h$. If $T_{\bar{1}}$ is invertible, define $T'_{\bar{1}} = (T_{\bar{1}})^{-1}$. Then 
\begin{align*}
    \text{Pr}[T_{\bar{1}} \text{ is invertible and }T'_{\bar{1}}T_{\bar{2}} \text{ is diagonalisable }] \leq \frac{d-2}{|S|}.
\end{align*}
\end{theorem}
\begin{proof}
As in the proof of Theorem \ref{thm:pdalgcomm}, we have
\begin{align*}
    T'_{\bar{1}}T_{\bar{2}} =   R^{-1}(\sum_{j_1...j_{d-2} \in [n]} (\prod_{m \in [d-2]}r_{j_m,2})(D_{\bar{1}})^{-1}S_{j_1...j_{d-2}})R
\end{align*}
where $D_{\bar{1}} = \sum_{j_1...j_{d-2} \in [n]}(\prod_{m \in [d-2]}r_{j_m,1})S_{j_1...j_{d-2}}$. So $T_{\bar{1}}'T_{\bar{2}}$ is diagonalisable iff $R$ is invertible and $M = (\sum_{j_1...j_{d-2} \in [n]} (\prod_{m \in [d-2]}r_{j_m,2})(D_{\bar{1}})^{-1}S_{j_1...j_{d-2}})$ is diagonalisable. We denote by $E_1$ the event  that $T_{\bar{1}}'T_{\bar{2}}$ is diagonalisable and by $E_1'$ the event that $M$ is diagonalisable.
\par
Let $E_2$ be the event  that $\{(D_{\bar{1}})^{-1}S_{i_1...i_{d-2}}\}_{i_1,...,i_{d-2} \in [n]}$ is a commuting family and there exists $j_1...j_{d-2} \in [n]$ such that $(D_{\bar{1}})^{-1}S_{j_1...j_{d-2}}$ is not diagonalisable.
Since $\mathcal{V}$ satisfies the commutativity property and does not satisfy the diagonalisability property, by Theorem \ref{thm:diagonalisable subspace}, the event  that $D_{\bar{1}}$ is invertible is the event same as $E_2$.
\par
Setting $A_{i_1...i_{d-2}} = (D_{\bar{1}})^{-1}S_{i_1...i_{d-2}}$ and setting $\alpha_{i} = r_{i,2}$ for all $i \in [n]$ and using Lemma \ref{lem:gendiag}, we get that 
\begin{equation}
    \text{Pr}_{R \in S}\big[E_1'\big|E_2\big] \leq \frac{d-2}{|S|} .
\end{equation}
%Note here we use the notation $\text{Pr}_{R \in S}$ to denote the probability when $r_{i,j}$'s are chosen uniformly and independently at random from $S$.
%
Now we know that $T_{\bar{1}}$ is invertible iff $R$ and $D_{\bar{1}}$ is invertible.
Let $E_3$ be the event  that  $T_{\bar{1}}$ is invertible and  $E_4$, the event  that $R$ is invertible. We have $E_3 = E_2 \cap E_4$ and we have seen that $E_1 = E_1' \cap E_4$. The probability of error can finally be bounded as follows:
\begin{align*}
    \text{Pr}_{R \in S}[E_1 \cap E_3] = \text{Pr}_{R \in S}[E_1' \cap E_2 \cap E_4]\leq \text{Pr}_{R \in S}[E_1'| E_2] \leq \frac{d-2}{|S|}.
\end{align*}
\end{proof}

\subsection{Analysis for negative inputs}\label{sec:negd}
In this section we complete the proof of Theorem \ref{thm:pdproof}. The case of positive inputs was treated in Section \ref{sec:positived}. It therefore remains to prove the following
result.
\begin{theorem}\label{thm:pdnegproof}
If an input $f \in \mathbb{K}[x_1,...,x_n]_d$ is not equivalent to some polynomial $P_d \in \mathcal{P}_d$, then $f$ is rejected by the algorithm with high probability over the choice of the random matrix $R$. More precisely, if the entries $r_{i,j}$ are chosen uniformly and independently at random from a finite set $S \subseteq \mathbb{K}$, then the input will be rejected with probability $\geq (1-\frac{2(d-2)}{|S|})$
\end{theorem}

\begin{proof}
Let $\{S_{i_1,...,i_{d-2}}\}_{i_1,...,i_{d-2} \in [n]}$ be the slices of $f$ and $\mathcal{V} =  \text{span}\{S_{\bar{1}},...,S_{\bar{n}}\}$. From Theorem \ref{thm:sodchar} and Theorem \ref{thm:main}, we know that if $f \not\sim P_d$, then there are three disjoint cases to consider
\begin{enumerate}
    \item \textbf{Case 1: }$\mathcal{V}$ is a weakly singular subspace of matrices.
    \item \textbf{Case 2: }$\mathcal{V}$ is not a weakly singular subspace and $\mathcal{V}$ does not satisfy the commutativity property.
    \item \textbf{Case 3: }$\mathcal{V}$ is not a weakly singular subspace, $\mathcal{V}$ satisfies the commutativity property but does not satisfy the diagonalisability property.
\end{enumerate}
Now we try to upper bound the probability of error in each case. 
\newline
In case 1, $T_{\bar{1}} = R^T(\sum_{j_1...j_{d-2} \in [n]} r_{j_1,1}...r_{j_{d-2},1}S_{j_1...j_{d-2}})R \in \mathcal{V}$ is always singular for any choice of $r_{j,1}$. So $f$ is rejected with probability $1$ in this case.
\newline
In case 2, we can upper bound the probability of error as follows: 
\begin{align*}
        &\text{Pr}_{R \in S}[f \text{ is accepted by the algorithm}] \\
        &=\text{Pr}_{R \in S}[T_{\bar{1}} \text{ is invertible, }T'_{\bar{1}}T_{\bar{2}}, T'_{\bar{1}}T_{\bar{3}} \text{ commute }, T'_{\bar{1}}T_{\bar{2}} \text{ is diagonalisable}] \\
        &\leq \text{Pr}_{R \in S}[T_{\bar{1}} \text{ is invertible, }T'_{\bar{1}}T_{\bar{2}}, T'_{\bar{1}}T_{\bar{3}} \text{ commute }] .
    \end{align*}
    Using Theorem \ref{thm:pdalgcomm}, we get that this occurs with probability at most $ \frac{2(d-2)}{|S|}$. In Case 3, we have the following upper bound on the probability of error,
    \begin{align*}
        &\text{Pr}_{R \in S}[f \text{ is accepted by the algorithm}] \\
        &=\text{Pr}_{R \in S}[T_{\bar{1}} \text{ is invertible, }T'_{\bar{1}}T_{\bar{2}}, T'_{\bar{1}}T_{\bar{3}} \text{ commute }, T'_{\bar{1}}T_{\bar{2}} \text{ is diagonalisable}] \\
        &\leq \text{Pr}_{R \in S}[T_{\bar{1}} \text{ is invertible, }T'_{\bar{1}}T_{\bar{2}} \text{ is diagonalisable}] .
    \end{align*}
    By Theorem \ref{thm:pdiagonalisability}, we can show that this occurs with probability $\leq \frac{d-2}{|S|}$. Therefore in all these three cases, the algorithm rejects $f$ with probability at least $ 1- \frac{2(d-2)}{|S|}$.
\end{proof}

\section{Variable Minimization}\label{sec:minvars}
We first recall the notion of redundant and essential variables studied by
Carlini \cite{10.1007/978-3-540-33275-6_15} and Kayal \cite{kayal11}.
\begin{definition}
A variable $x_i$ in a polynomial $f (x_1,...,x_n )$ is redundant if $f$
does not depend on $x_i$ , i.e., $x_i$ does not appear in any monomial of $f$.
\par
Let $f \in \mathbb{K}[x_1,...,x_n]$. The number of essential variables is the smallest number $t$ such that there exists an invertible linear transformation $A \in \mathbb{K}^{n \times n}$ on the variables such that every monomial of $f(Ax)$ contains only the variables $x_1,...,x_t$.
\end{definition}
In this section we propose the following algorithm for variable minimization.

\begin{algorithm}[H]\label{alg:minvars}
\SetAlgoLined
%\KwResult{Write here the result }
 \textbf{Input:} A degree-$d$ homogeneous polynomial $P$ given by a blackbox
 \\
% \While{While condition}
Pick $(\alpha_1,...,\alpha_n)$ where $\alpha_j = (\alpha_j^{(1)},...,\alpha_j^{(n)})$ and $\alpha_j^{(i)}$ are picked uniformly and independently at random from a finite set $S \subset \mathbb{K}$ \\
Compute $M = (\frac{\partial P}{\partial x_j}(\alpha_i))$, such that $i,j \in [n]$ \\
Perform Gaussian elimination on $M$ and define the basis of the kernel $B = \{v_1,...,v_{n-t}\}$ \\
Add vectors $u_1,...,u_{t}$ to $B$ to obtain a basis for $\mathbb{K}^n$ \\
Define $n \times n$ matrix $A = (u_1,...,u_{t},v_1,...,v_{n-t})$ where $u_i$ and $v_j$ are the columns of $A$\\
Let $f(x) = P(Ax)$ \\
Run Algorithm \ref{alg:pd} on $f(x_1,...,x_t,0,...,0)$ \\
\eIf{Algorithm \ref{alg:pd} accepts}{accept}{reject}
 \caption{Randomized algorithm for variable minimization}
\end{algorithm}
A randomized algorithm for minimizing the number of variables is given
in (\cite{kayal11}, Theorem 4.1). More precisely, if the input $f$ has $t$ essential variables the algorithm finds (with high probability) an invertible matrix $A$ such that $f(Ax)$ depends on its first $t$ variables only. It is based on the observation that $t= \text{dim}(\partial f)$ where $\partial(f)$ denotes the tuple of $n$ first-order partial derivatives $\frac{\partial f}{\partial x_i}$ (and $\text{dim}(\partial f)$ denotes the dimension of the spanned subspace). \cite{kayal11} then uses Lemma \ref{lem:matrixminvars} along with Theorem \ref{thm:carlini} to return the required invertible matrix.
\par
We combine the algorithm for minimizing the number of variables along with Algorithm \ref{alg:pd} to check if there exists a decomposition of the polynomial into linear combination of $d$-th powers of $\leq n$ many linearly independent linear forms. More formally, Algorithm \ref{alg:minvars} decides in polynomial time over $\mathbb{C}$ or $\mathbb{R}$ whether the input $f$ which is given as blackbox can be written as $\sum_{i=1}^t \alpha_i l_i^d$ for some $t \leq n$ where $l_i$'s are linearly independent linear forms and $\alpha_i \neq 0$ for all $i \in [t]$. We do a detailed complexity analysis of this algorithm in Appendix \ref{app:minvars}.
\begin{definition}
Let $\textbf{f}(x) = (f_1(x),...,f_m(x)) \in (\mathbb{K}[x])^m$ be a vector of polynomials over a field $\mathbb{K}$. The set of $\mathbb{K}$-linear dependencies in $\textbf{f}$, denoted by $\textbf{f}^{\perp}$, is the set of all vectors $\textbf{v} \in \mathbb{K}^m$, whose inner product with $\mathbf{f}$ is the zero polynomial i.e
\begin{align*}
    \textbf{f}^{\perp} := \{(a_1,...,a_m) \in \mathbb{K}^m| \sum_{i \in [m]} a_if_i(x) = 0\}
\end{align*}
\end{definition}
This following lemma from \cite{kayal11} gives a randomized algorithm to compute the basis of linear dependencies of a vector of polynomials. We restate it here for completeness and also calculate the probability bounds which we will need for our correctness proof of Algorithm \ref{alg:minvars}.
\begin{lemma}\label{lem:matrixminvars}
Let $f = (f_1(X),...,f_m(X))$ be a vector of $m$ polynomial with $\text{deg}(f_i) \leq d$ such that $\text{rank}(f^{\perp}) = t$. Pick $a_1,...,a_m$ where $a_j = (a_j^{(1)},...,a_j^{(m)})$ where the $a_j^{(i)}$'s are chosen uniformly and independently at random from a finite set $S \subseteq \mathbb{K}$ for all $i,j \in [m]$. Define matrix
\begin{align*}
 P(a_1,...,a_m)  = \begin{bmatrix} 
    f_1(a_1) & f_2(a_1) & \dots & f_m(a_1) \\
    f_1(a_2) & f_2(a_2) & \dots & f_m(a_2) \\
    \vdots &  \vdots & \ddots & \vdots\\
    f_1(a_m) & f_2(a_m) & \dots & f_m(a_m) 
    \end{bmatrix}
\end{align*}
Then 
\begin{align*}
    \text{Pr}_{a \in S}[\text{rank}(P(a_1,...,a_m)) = m-t] \geq  1 - \frac{(m-t)d}{|S|}
\end{align*}
Additionally, if $\text{rank}(P(a_1,...,a_m)) = m-t$, then $\text{ker}(P(a_1,...,a_m)) = f^{\perp}$. 
\end{lemma}
\begin{proof}
Without loss of generality, we assume that the polynomial $f_1,...,f_{m-t}$ are $\mathbb{K}$-linearly independent and the rest of the polynomials are $\mathbb{K}$ linear combinations of the first $m-t$ polynomials. So it is enough to prove that
\begin{align*}
     Q(a_1,...,a_m)  = \begin{bmatrix} 
    f_1(a_1) & f_2(a_1) & \dots & f_{m-t}(a_1) \\
    f_1(a_2) & f_2(a_2) & \dots & f_m(a_2) \\
    \vdots &  \vdots & \ddots & \vdots\\
    f_1(a_{m-t}) & f_2(a_{m-t}) & \dots & f_m(a_{m-t}) 
    \end{bmatrix}
\end{align*}
has full rank, which is equivalent to proving that $\text{det}(Q)(a_1,...,a_m) \neq 0$.
\newline
Now $\text{deg}(\text{det}(Q)(x_1,...,x_n)) \leq (m-t)d$.  Claim 7 in \cite{kayal11} shows that $\text{det}(Q) \not\equiv 0$. Applying Schwartz-Zippel lemma, we get that,
\begin{align*}
    \text{Pr}_{a \in S}[\text{rank}(P(a_1,...,a_m)) = m-t] \geq 1 - \frac{(m-t)d}{|S|}
\end{align*}
\end{proof}
Recall that we define $\partial(f) = (\frac{\partial f}{\partial x_1},...,\frac{\partial f}{\partial x_n})$. Let $b_1,...,b_{n-t}$ be a basis for $\partial(f)^{\perp}$. Now there exists $t$ independent vectors $a_1 ,..., a_{t}$ such that the vector space $\mathbb{K}^n$ is spanned by $a_1,...,a_{t},b_1,..., b_{n-t}$. We define $A_f$ to be the invertible matrix whose columns are $a_1,...,a_{t},b_1,..., b_{n-t}$.
\begin{theorem}\cite{10.1007/978-3-540-33275-6_15}\label{thm:carlini}
The number of redundant variables in a polynomial $f (x)$ equals the dimension
of $\partial(f)^{\perp}$. Furthermore, given a basis of $\partial(f)^{\perp}$ , the polynomial $f (A_fX)$ depends on only the first $(n - \text{dim}(\partial(f)^{\perp}) ))$ variables.
\end{theorem}
\begin{lemma}\label{lem:linformsessvarseq}
If $f$ can be written as a sum of $r$ powers of linearly independent linear forms, then the number of essential variables of $f$ is equal to $r$.
\end{lemma}
\begin{proof}
Follows from Example 42 in \cite{koiran2020derandomization}.
\end{proof}
We combine all the results in this section along with Theorem \ref{thm:pdproof} to give a correctness proof for Algorithm \ref{alg:minvars}. 
\begin{theorem}
If an input $P \in \mathbb{K}[x_1,...,x_n]_d$ can not be written as $\sum_{i=1}^t \alpha_il_i^d$, where $l_i$ are linearly independent linear forms and $\alpha_i \neq 0$ for all $i \in [t]$, for any $t \leq n$, then $P$ is rejected by Algorithm \ref{alg:minvars} with high probability over the choice of the random matrix $R$ and the points $\alpha_1,...,\alpha_n$. More formally, if the entries $\alpha_j^i$ and $r_{i,j}$ are chosen uniformly and independently at random from a finite set $S$, then the input will be rejected with probability $\geq (1-\frac{2(d-2)}{|S|})(1-\frac{n(d-1)}{|S|})$.
\newline
Conversely, if $P \in \mathbb{K}[x_1,...,x_n]_d$ can be written as $\sum_{i=1}^t \alpha_il_i^d$ where $l_i$ are linearly independent linear forms and $\alpha_i \neq 0$ for all $i \in [t]$, then $P$ will be accepted by Algorithm \ref{alg:minvars}
with high probability over the choice of the random matrix $R$ and the points $a_1,...,a_n$. More formally, if the entries $\alpha_j^{(i)}$ and $r_{i,j}$  are chosen uniformly and independently at random from a finite set $S$, then the input will be accepted with probability $\geq (1-\frac{t(d-1)}{|S|})^2$ .
\end{theorem}
\begin{proof}
Let us assume that $P$ has $t$ essential variables and we fix $M$ as given by the algorithm. Then using Lemma \ref{lem:matrixminvars} for $f_i = \frac{\partial P}{\partial x_i}$ and $a_i = \alpha_i$, we get that 
\begin{align*}
    \text{Pr}_{\alpha}[v_1,...,v_{n-t} \text{ is a basis for } \partial(P)^{\perp}] \geq 1 - \frac{t(d-1)}{|S|}.
\end{align*}
The linear transformation $A$ defined in the algorithm, satisfies the conditions of the linear transformation defined in Theorem \ref{thm:carlini} with $b_i = v_i$ and $a_i = u_i$ with probability $\geq 1 - \frac{t(d-1)}{|S|}$. Now, we define $f(x) = P(Ax)$. This gives us that
\begin{align*}
    \text{Pr}_{\alpha}[f(x) \text{ depends only on the first } t \text{ variables} ] &\geq 1 - \frac{t(d-1)}{|S|} \\
    &\geq 1 - \frac{n(d-1)}{|S|}.
\end{align*}
Using Theorem \ref{thm:pdproof}, we get that if $f$ can not be written as a sum of $t$-many sum of $d$-th powers of linear forms for any $t \leq n$, then the algorithm rejects the polynomial with probability $\geq (1-\frac{2(d-2)}{|S|})(1-\frac{n(d-1)}{|S|})$. 
\par
For the converse, if $f$ can be written as a sum of $t$-many sum of $d$-th powers of linear forms, using Lemma \ref{lem:linformsessvarseq}, then the polynomial has $t$ essential variables. Then the algorithm accepts $P$ with probability $\geq (1-\frac{t(d-1)}{|S|})^2$.
\end{proof}

\section{Reconstruction Algorithm for $P_d$}\label{sec:reconstruction}
The general reconstruction problem for a special class of arithmetic circuits can be stated as follows: Given a homogeneous degree-$d$ polynomial, output the smallest circuit that computes it. In this section, we look at the reconstruction problem for linear combination of $d$-th powers of linearly independent linear forms. Notice that Algorithm \ref{alg:pd} already solves the decision version of this problem in polynomial time i.e. Given a homogeneous degree-$d$ polynomial, can it be written as a linear combination of $d$-th powers of linearly independent linear forms?
\par
In terms of the polynomial equivalence problem, we have already shown that given a homogeneous degree-$d$ polynomial $f$, we can check in polynomial time if there exists an invertible matrix $A$ and some $P_d \in \mathcal{P}_d$ such that $f(x) = P_d(Ax)$. We give an algorithm (see Algorithm \ref{alg:reconstruction} below) that uses Algorithm \ref{alg:pd} to check the existence of such an $A$ and then outputs the $A$ and the corresponding $P_d$. When $f \in \mathbb{C}[x_1,...,x_n]_d$, this algorithm runs in polynomial time, if we allow the computation of polynomial roots in our model.

\begin{algorithm}[H]\label{alg:reconstruction}\label{alg:reconstruction}
\SetAlgoLined
\KwResult{The algorithm checks if $f$ is equivalent to some polynomial in $ \mathcal{P}_d$ and outputs $\{(\alpha_1,l_1),...,(\alpha_n,l_n)\}$ where $0 \neq \alpha_i \in \mathbb{K}$ and $l_i$'s are linearly independent linear forms such that $f= \sum_{i=1}^n \alpha_il_i^d$}
 \textbf{Input:} A degree-d homogeneous polynomial $f$ given as blackbox\\
% \While{While condition}{
Let $R \in M_n(\mathbb{K})$ be a matrix such that its entries $r_{ij}$ are picked uniformly and independently at random from a finite set $S$ and set $h(x) = f(Rx)$  \\
Let $T_{i_1,...,i_{d-2}}$ be the slices of $h$ for all $i_1,...,i_{d-2} \in [n]$ \\
Compute $T_{\bar{1}}, T_{\bar{2}}, T_{\bar{3}}$ \\
\eIf{$T_{\bar{1}}$ is singular}{reject}{compute $T_{\bar{1}}' = T_{\bar{1}}^{-1}$ \\
   \eIf{ $T_{\bar{1}}'T_{\bar{2}}$ and $T_{\bar{1}}'T_{\bar{3}} $ commute and $T_{\bar{1}}'T_{\bar{2}}$ is diagonalisable over $\mathbb{K}$}{
   diagonalise $T'_{\bar{1}}T_{\bar{2}} = P\Lambda P^{-1}$ \\
   Let $l_i$ be the $i$-th row of $(R^{-1}P^{-1})$ and let $\alpha_i = f(PRx)(e_i)$ where $e_i \in \mathbb{K}^n$ is the $i$-th standard basis vector for all $i \in [n]$ \\
   Output $\{(\alpha_1,l_1),...,(\alpha_n,l_n)\}$
   }{
   reject
  }
  }
 \caption{Randomized Reconstruction Algorithm}
\end{algorithm}
{ Note that this output is unique up to the permutation and scaling of the linear forms by a constant. If the linear form $l_i$ is scaled by a constant $c$, then it is reflected in the $\alpha_i$ which becomes $\frac{\alpha_i}{c^d}$}.
\par
A natural question is to study the approximate version for this reconstruction problem i.e. if the input polynomial admits such a decomposition, the linear forms which are returned by the algorithm are "arbitrarily close" to the required linear forms. In this context, one should note that interestingly, the only non-algebraic step in this algorithm is the step of matrix diagonalization. All other steps can be computed exactly in polynomial time. Recently, \cite{BGKS20} gave a poly-time randomized algorithm for the approximate version of matrix diagonalization problem. Referring to that algorithm for our diagonalization step combined with our reconstruction algorithm should effectively give a polynomial time randomized algorithm for the approximate version of the reconstruction problem. A more precise analysis is left for future work.
\par
For the next lemma, we assume that the input polynomial $f$ is equivalent to some polynomial $ P_d \in \mathcal{P}_d$ i.e. $f(x) = P_d(Ax)$ where $A$ is an invertible matrix. Take a random change of variables and let $T_{\bar{1}}$ and $T_{\bar{2}}$ be two slices of the new polynomial. Then, the eigenvalues of $(T_{\bar{1}})^{-1}T_{\bar{2}}$ are distinct with high probability over the random change of variables. Lemma \ref{thm:eigenvalues} along with Corollary \ref{corr:reconstruction} ensures that just diagonalization of a single $(T_{\bar{1}})^{-1}T_{\bar{2}}$ is enough to recover the matrix $A$ uniquely (up to permutation and scaling of the rows).
\begin{theorem}\label{thm:eigenvalues}
Let $f(x) = P_d(Ax)$ for some $P_d \in \mathcal{P}_d$ where $A$ is an invertible matrix. Let $h(x) = f(Rx)$ where the entries $r_{ij}$ of $R$ are chosen uniformly and independently at random from a finite set $S$. Let the slices of $h$ be $\{T_{i_1...i_{d-2}}\}_{i_1,...,i_{d-2} \in [n]}$. If $T_{\bar{1}}$ is invertible, let $T'_{\bar{1}} = T_{\bar{1}}^{-1}$. Let 
$\lambda_1,...,\lambda_n$ be the eigenvalues of $T'_{\bar{1}}T_{\bar{2}}$. Then
\begin{align*}
    &\text{Pr}_{R \in S}[\text{There exists }i,j \in [n] \text{ such that }\lambda_i = \lambda_j \text{ and } T_{\bar{1}} \text{ is invertible }] \\
    &\leq \frac{2\binom{n}{2}(d-2)}{|S|}.
\end{align*}
\begin{proof}
Let us assume that $P_d(x) = \sum_{i=1}^n \alpha_i x_i^d$ where $\alpha_i \neq 0$. We use the fact that $h(x) = P_d(RAx)$. Let $\{S_{i_1...i_{d-2}}\}_{i_1,...,i_{d-2} \in [n]}$ be the slices of $P_d$. Then we know from \eqref{eq:comm}, 
\begin{align*}
    T'_{\bar{1}}T_{\bar{2}} = (RA)^{-1}(\sum_{j_1...j_{d-2} \in [n]} \prod_{m \in [d-2]}(RA)_{j_m,2}(D_{\bar{1}})^{-1}S_{j_1...j_{d-2}})(RA)
\end{align*}
where 
\begin{align*}
    D_{\bar{1}} = \sum_{j_1...j_{d-2} \in [n]} (\prod_{m \in [d-2]}(RA)_{j_m,1})S_{j_1...j_{d-2}}.
\end{align*}
Since $S_{i_1...i_{d-2}}$ are the slices of $P_d$, we know that 
\begin{align*}
    S_{i_1...i_{d-2}} &= \alpha_i\text{diag}(e_i) \text{ if } i_1 = i_2 = ... = i_{d-2} = i \\
    &= 0 \text{ otherwise. }
\end{align*}
Since $T_{\bar{1}}$ is invertible, then $R$ and $D_{\bar{1}}$ are invertible. Now
\begin{align*}
    D_{\bar{1}} = \text{diag}(\alpha_1((RA)_{1,1})^{d-2},...,\alpha_n((RA)_{n,1})^{d-2}).
\end{align*}
Since $D_{\bar{1}}$ is invertible, this gives us that $\alpha_i((RA)_{i,1})^{d-2} \neq 0$ for all $i \in [n]$ and we obtain
\begin{equation}\label{eq:reconeigenvalues}
\begin{split}
    T'_{\bar{1}}T_{\bar{2}} = (RA)^{-1}\Big(\text{diag}(\frac{((RA)_{1,2})^{d-2}}{((RA)_{1,1})^{d-2}},...,\frac{((RA)_{n,2})^{d-2}}{((RA)_{n,1})^{d-2}})\Big)RA
\end{split}
\end{equation}
This gives us that $\lambda_i = \frac{((RA)_{i,2})^{d-2}}{((RA)_{i,1})^{d-2}}$. We define 
\begin{align*}
    P_{i,j}(R) = ((RA)_{i,2}(RA)_{j,1})^{d-2} - ((RA)_{j,2}(RA)_{i,1})^{d-2}.
\end{align*}
We can see that $\lambda_i \neq \lambda_j$ iff $P_{i,j}(R) \neq 0$. Also, $\text{deg}(P_{ij}) \leq 2(d-2)$. We can choose $R$ such that $(RA)_{i,2} = 1$, $(RA)_{j,1} = 1$, $(RA)_{i,1} = 0$, $(RA)_{j,2} = 0$, (If $A$ is invertible, there exists $R$ such that $RA = M$ for any matrix $M$). Hence, $P_{ij} \not\equiv 0$. Using Schwartz- Zippel Lemma, we get that
\begin{align*}
    \text{Pr}_{R \in S}[P_{ij} = 0] \leq \frac{2(d-2)}{|S|}.
\end{align*}
This gives us that
\begin{align*}
    \text{Pr}_{R \in S}[T_{\bar{1}} \text{ is invertible and }\lambda_i = \lambda_j] 
    \leq \text{Pr}_{R \in S}[P_{ij} = 0] \leq \frac{2(d-2)}{|S|}.
\end{align*}
Taking the union bound over all possible pairs of $i,j \in [n]$, we get that
\begin{align*}
    &\text{Pr}_{R \in S}[\text{There exists }i,j \in [n] \text{ such that }\lambda_i = \lambda_j \text{ and } T_{\bar{1}} \text{ is invertible }] \\
    &\leq \frac{2\binom{n}{2}(d-2)}{|S|}.
\end{align*}
\end{proof}
\end{theorem}

\begin{corollary}\label{corr:reconstruction}
Let $f(x)$ be a degree-$d$ form which is equivalent to some polynomial in $\mathcal{P}_d$. Let $h(x) = f(Rx)$ where the entries $r_{ij}$ of $R$ are chosen uniformly and independently at random from a finite set $S \subset \mathbb{K}$. Let the slices of $h$ be $\{T_{i_1,...,i_{d-2}}\}_{i_1,...,i_{d-2} \in [n]}$. If $T_{\bar{1}}$ is invertible, let $T'_{\bar{1}} = T_{\bar{1}}^{-1}$. Suppose $T'_{\bar{1}}T_{\bar{2}}$ can be diagonalised as $P \Lambda P^{-1}$. Let $l_i$ be the rows of $R^{-1}P^{-1}$. Define $\alpha_i = f(PRx)(e_i)$ where $e_i \in \mathbb{K}^n$ is the $i$-th standard basis vector for all $i \in [n]$. Then 
\begin{align*}
    \text{Pr}_{R \in S}[T_{\bar{1}} \text{ is invertible and }f(x) = \sum_{i=1}^n \alpha_il_i^d] \geq 1 - \frac{2\binom{n}{2}(d-2)}{|S|}.
\end{align*}
\end{corollary}
\begin{proof}
By Theorem \ref{thm:eigenvalues}, we get that
\begin{align*}
    &\text{Pr}_{R \in S}[T_{\bar{1}} \text{ is invertible and } T'_{\bar{1}}T_{\bar{2}}\text{ has distinct eigenvalues}] \\
    &\geq 1 - \frac{2\binom{n}{2}(d-2)}{|S|}.
\end{align*}
We will now show that if $T_{\bar{1}} \text{ is invertible and } T'_{\bar{1}}T_{\bar{2}}\text{ has distinct eigenvalues}$, then $f(x) = \sum_{i=1}^n \alpha_i l_i^d$.
\par
If the eigenvalues are distinct, then the rank of the eigenspaces corresponding to each eigenvalue is $1$. Hence, the eigenvectors of $T_{\bar{1}}'T_{\bar{2}}$ are unique (up to a scaling factor). We already know that $h(x) = P_d(Bx)$ for some $ P_d \in \mathcal{P}_d$ such that $B$ is invertible. Then the columns of $B^{-1}$ form the eigenvectors of $T_{\bar{1}}'T_{\bar{2}}$. We take the diagonalization of $T_{\bar{1}}'T_{\bar{2}}$ into $P\Lambda P^{-1}$. Note here that the columns of $P$ form the eigenvectors for $T_{\bar{1}}'T_{\bar{2}}$. The uniqueness of eigenvectors of $T_{\bar{1}}'T_{\bar{2}}$ gives us that the set of columns of $P$ are essentially the  set of columns of $B^{-1}$ upto a scaling factor. So this gives us that $h(x) = P_{d}'(P^{-1}x)$ for some $P_{d}' \in \mathcal{P}_d$. We know that $f(x) = h(R^{-1}x)$. This gives us that $f(x) = P_d'(R^{-1}P^{-1})x$.
We define $A = R^{-1}P^{-1}$ and the $i$-th row of $A$ as $l_i$. This fixes the set of linear forms of the decomposition of the input polynomial which are unique up to a scaling factor. 
\newline
Taking $P_d = f(A^{-1}x)$ gives the corresponding polynomial in $\mathcal{P}_d$ such that $f$ is equivalent to $P_d$. Now $P_d = \sum_{i=1}^n \alpha_i x_i^d$ where $\alpha_i \neq 0$. Evaluating $P_d$ at $e_i \in \mathbb{K}^n$ where $e_i$ is the $i$-th standard basis vector, returns the corresponding $\alpha_i$. Hence,
\begin{align*}
    &\text{Pr}_{R \in S}[f(x) = \sum_{i=1}^n \alpha_il_i^d \text{ and } T_{\bar{1}} \text{ is invertible}] \geq 1 - \frac{2\binom{n}{2}(d-2)}{|S|}.
\end{align*}
\end{proof}
We combine all the results in this section along with Theorem \ref{thm:pdproof} to give a correctness proof of Algorithm \ref{alg:reconstruction}.
\begin{theorem}
If an input $f \in \mathbb{K}[x_1,...,x_n]_d$ is not equivalent to some polynomial in $\mathcal{P}_d$, then $f$ is rejected by Algorithm \ref{alg:reconstruction} with high probability. More formally, if the entries $r_{i,j}$ of $R$ are chosen uniformly and independently at random from a finite set $S$, then the input will be rejected with probability $\geq (1 - \frac{2(d-2)}{|S|})$.
\newline
Conversely, if an input $f \in \mathbb{K}[x_1,...,x_n]_d$ is equivalent to some polynomial in $\mathcal{P}_d$, then Algorithm \ref{alg:reconstruction}  outputs such a polynomial with high probability. More formally,if the entries of a matrix $R$ are chosen uniformly and independently from a finite set $S$, then the algorithm outputs a set of linearly independent linear forms $l_i$ and corresponding coefficients $\alpha_i \neq 0$ such that $f = \sum_{i=1}^n \alpha_il_i^d$ with probability  $\geq 1 - (\frac{2\binom{n}{2}(d-2)}{|S|} + \frac{n(d-1)}{|S|})$.
\end{theorem}
\begin{proof}
From Theorem \ref{thm:pdnegproof}, we get that if $f$ is not equivalent to any polynomial in $\mathcal{P}_d$, then $f$ is rejected by Algorithm \ref{alg:reconstruction} with probability $\geq (1 - \frac{2(d-2)}{|S|})$.
\par
For the converse, we start by assuming that $f$ is equivalent to some polynomial in $\mathcal{P}_d$. We know that if the first slice $T_{\bar{1}}$ of $h(x) = f(Rx)$ is not invertible, the Algorithm always makes an error and rejects the input. From Theorem \ref{sodcharspl}, we know that the subspace spanned by the slices of $f$ is not weakly singular. We can therefore apply Lemma~\ref{lem:pdt1invertible}, we get that
\begin{equation} \label{eq:recont1notinv}
    \text{Pr}_{R \in S}[T_{\bar{1}} \text{ is not invertible}] \leq \frac{n(d-1)}{|S|}. 
\end{equation}

\par
Moreover if $T_{\bar{1}}$ is invertible, Lemma \ref{lem:pdopp} shows that $f$ will always be accepted. Let the output of the algorithm be $\{(\alpha_i,l_i)\}_{i \in [n]}$. So the only possible error is when $f \neq \sum_{i \in [n]} \alpha_i l_i^d$.
From Corollary \ref{corr:reconstruction}, we get that
\begin{equation}\label{eq:recont1invertible}
    \text{Pr}_{R \in S}[\text{Algorithm makes an error and } T_{\bar{1}} \text{ is invertible}] \leq \frac{2\binom{n}{2}(d-2)}{|S|}.
\end{equation}
Combining (\ref{eq:recont1notinv}) and (\ref{eq:recont1invertible}), we get that if $f$ is equivalent to some polynomial in $\mathcal{P}_d$, then Algorithm \ref{alg:reconstruction} returns a set of linearly independent linear forms $l_i$ and corresponding coefficients $\alpha_i \neq 0$ (which are unique up to scaling and permutation) such that $f = \sum_{i=1}^n \alpha_il_i^d$ with probability $\geq 1 - (\frac{2\binom{n}{2}(d-2)}{|S|} + \frac{n(d-1)}{|S|})$.
\end{proof}
We can also replace the call to Algorithm \ref{alg:pd} in Algorithm \ref{alg:minvars} by a call to Algorithm \ref{alg:reconstruction} to similarly get a reconstruction algorithm for linear combination of powers of at most $n$ linearly independent linear forms. More specifically, given a polynomial $f$ in blackbox, it will check if there exists a decomposition of $f = \sum_{i=1}^t \alpha_i l_i^d$ for some $t \leq n$ where $l_i$'s are linearly independent and $\alpha_i \neq 0$ and outputs the decomposition, if it exists.
\appendix
\section{Appendix: Complexity analysis for equivalence to a sum of cubes} \label{app:cubes}
We first explain how the diagonalizability of a matrix can be tested efficiently with an algebraic algorithm. This can be done thanks to the following classical result from linear algebra (see e.g. \cite{horn13}, Corollary 3.3.8 for the case $\mathbb{K} = \mathbb{C}$).
\begin{lemma}\label{lem:compdiag}
Let $\mathbb{K}$ be a field of characteristic $0$ and let $\chi_M$ be the characteristic polynomial of a matrix $M \in M_n(\mathbb{K})$. Let $P_M = \frac{\chi_M}{\text{gcd}(\chi_M,\chi_M' )}$ be the square-free part of $\chi_M$. The matrix $M$ is diagonalisable over $\bar{\mathbb{K}}$ iff
$P_M(M) = 0$. Moreover, in this case $M$ is diagonalisable over $\mathbb{K}$ iff all the
roots of $P_M$ lie in  $\mathbb{K}$.
\end{lemma}
We prove the following lemma which shows that each entry of the slices that we need to compute can be computed using $O(d)$ calls to the blackbox and $O(d\log^3 d)$ many arithmetic operations. This proof is motivated from the idea of polynomial interpolation and the proof strategy of Lemma 4 in \cite{forbes_et_al:LIPIcs:2018:9058}.Their algorithm gives a $\text{poly}(sd)$ runtime in our setting. In this section, we will require this lemma only for $d = 3$, but we prove the general form so that we can use it later in Appendix \ref{app:dthpowers}.
\begin{lemma}\label{lem:circuitcomplexity}
Let $f \in \mathbb{K}[x_1,...,x_n]$ be a homogeneous polynomial of degree $d$ where %PK $|\marF| > d$. 
$|\mathbb{K}| > d$. 
If $f$ is input as a blackbox $C$, then for some $i \in [n]$ can compute the coefficient of $x_i^{d-2}x_kx_j$ using $O(d)$ many oracle calls to the blackbox and $O(M(d)\log d)$ many arithmetic operations.
\end{lemma}
\begin{proof}
Here we use the standard trick of polynomial interpolation. Without loss of generality, we assume that $i=1$, that is we need to compute $\text{coeff}_{x_1^{d-2}x_jx_k}(f)$. So we can write,
\begin{align*}
    C(x_1,...,x_n) = \sum_{i = 0}^d c_ix_1^i \text{ where } c_j \in \mathbb{K}[x_2,...,x_n].
\end{align*}
Now there are three cases:
\begin{itemize}
    \item $j = k = 1$
    \item only one of $j$ or $k = 1$
    \item $j,k \neq 1$.
\end{itemize}
\textbf{Case 1:}
Evaluate the polynomial at the point $(1,0,...,0) \in \mathbb{K}^n$. This gives us the coefficient of $x_1^d$ in $f$.
\newline
\textbf{Case 2:} Exactly one of $j$ or $k$ is $1$. Without loss of generality, we assume $j = 1$ and $k=2$. So we want to compute $\text{coeff}_{x_1^{d-1}x_2}(f)$. We evaluate the polynomial $f$ at the point $\bar{t} = (t,1,0,...,0) \in \mathbb{K}^n$. Now it's easy to check that
$\text{coeff}_{t^{d-1}}(f(\bar{t})) = \text{coeff}_{x_1^{d-1}x_2}(f)$. So we need to only interpolate and calculate the coefficient of $t^{d-1}$ in $f(\bar{t})$. 

\textbf{Case 3:} $j,k \neq 1$.
\newline
Now in this, there are two cases : 
\begin{itemize}
    \item $j = k = 2$ : Here we take a similar strategy as Case 2. We evaluate $f$ at $\bar{t}  = (t,1,0,...,0) \in \mathbb{K}^n$. Then $\text{coeff}_{t^{d-2}}(f(\bar{t})) = \text{coeff}_{x_1^{d-2}x_2^2}(f)$
    So interpolate and compute the coefficient of $t^{d-2}$ in $f(\bar{t})$.
    \item The final case is when the indices are all distinct. Let $j = 2 $ and $k = 3$ without loss of generality. We evaluate the polynomial $f$ at $\bar{t} = (t,1,1,0,...,0) \in \mathbb{K}^n$. 
    Now
    \begin{align*}
        \text{coeff}_{t^{d-2}}(f(\bar{t})) = \text{coeff}_{x_1^{d-2}x_2^2}(f) + \text{coeff}_{x_1^{d-2}x_2x_3}(f) + \text{coeff}_{x_1^{d-2}x_3^2}(f) 
    \end{align*}
    Now, using the previous case, we compute $\text{coeff}_{x_1^{d-2}x_2^2}(f)$ and $\text{coeff}_{x_1^{d-2}x_3^2}(f)$, subtract them from $\text{coeff}_{t^{d-2}}(f(\bar{t}))$ and return the answer.
\end{itemize}
Each case of this algorithm requires us to do univariate polynomial interpolation at most constantly many number of times and this can be done in time $O(M(d)\log d)$  \cite{10.5555/2512973} (Section 10.2)  where $M(d)$ is the number of arithmetic operations for polynomial multiplication. Rest of the operations can be done in time $O(1)$. It also requires $O(d)$ many oracle calls to the blackbox.
\newline
So the algorithm uses $O(M(d)\log d)$ many arithmetic operations and $O(d)$ many oracle calls to the blackbox.
\end{proof}

\begin{theorem}\label{thm:complexitycubes}
If a degree $3$ form $f \in \mathbb{K}[x_1,...,x_n]$ is given in dense representation, Algorithm \ref{algo:cubes} runs in time $O(n^{\omega+1})$ where $\omega$ is the exponent of matrix multiplication.
\newline
If the degree $3$ form $f \in \mathbb{K}[x_1,...,x_n]$ is given as a blackbox then the algorithm makes $O(n^2)$ many calls to the blackbox and $O(n^{\omega+1})$ many arithmetic operations.
\end{theorem}
\begin{proof}
The following are the different stages of computation required in this algorithm:
\begin{enumerate}
    \item Recall from Theorem \ref{thm:P3structural}, the slices $T_i$ of $h = f(Rx)$ are given by the formula $T_k = R^T(\sum_{i\in [n]}r_{i,k}S_i)R$. If the polynomial is input in dense representation, then the elements of $S_i$ can be computed from the coefficients of $f$. Then we take linear combinations of the $S_i$'s and computing $T_1,T_2,T_3$ takes $O(n^3)$ many arithmetic operations.
    \newline
    If the polynomial is given as a blackbox, we compute $x' = Rx$ and we call the blackbox on this input.  
    \item \textbf{Compute } $T_{1}$, $T_{2}$, $T_{3}$
        We know 
        \begin{align*}
            (T_{k})_{ij} = \frac{1}{3!}\partial_{x_ix_jx_k}(h) 
        \end{align*}
         So we can extract each entry of $T_{k}$ using constant many calls to the blackbox and constantly many arithmetic operations using Lemma~\ref{lem:circuitcomplexity}. There are in total $3n^2$ such entries that we need to compute. So the total number of calls to the blackbox is $O(n^2)$ and the number of arithmetic operations is $O(n^2)$.
    \item \textbf{Check if }$T_1$\textbf{ is invertible. If invertible, compute }
    $T_1' = T_1^{-1}$.
    \newline
    This can be done in time at most $O(n^3)$. (Faster algorithms exist \cite{10.5555/2512973} but this bound is enough since it is not the most expensive step of the algorithm.)
    \item \textbf{Checking commutativity of }$T_1'T_2$ and $T_1'T_3$.
    \newline
    Here we compute the product $T_1'T_2T_1'T_3$ and $T_1'T_3T_1'T_2$ and check if their difference is $0$. This can be done in time $O(n^{\omega})$. 
    \item \textbf{Checking the diagonalisability of }$T_1'T_2$:
\newline
Here we use Lemma \ref{lem:compdiag}. Hence there are four steps:
\begin{itemize}
    \item Compute the characteristic polynomial of $M$ i.e. $\chi_M$. This can be done in time $O(n^{\omega}\log n)$. \cite{KELLERGEHRIG1985309}
    \item Compute gcd$(\chi_M,\chi_M',)$. Since, $deg(\chi_M) \leq n$, this can be done in $O(n^2)$ using Euclidean Algorithm. \cite{10.5555/2512973}
    \item Compute $P_M = \frac{\chi_M}{\text{gcd}(\chi_M,\chi_M' )}$. This is can be computed in $O(n^2)$ using the standard long-division algorithm.
    \item Check if $P_M(M) = 0$. Using Horner's Method, we can evaluate the polynomial at $M$ using $n$ many matrix multiplications only. Hence, computing $P_M(M)$ takes $O(n^{\omega+1})$ time. 
\end{itemize}
Hence, we can conclude that the diagonalisability of $T_1'T_2$ can be checked in time $O(n^{\omega+1})$. Note that this is the most expensive step of the algorithm!
\end{enumerate}
So we conclude that if the polynomial is given as an input in the dense representation model, then the algorithm runs in time $O(n^{\omega+1})$. 
\newline
If the polynomial is given as a blackbox, then the algorithm makes $O(n^2)$ many oracle calls to the blackbox and takes $O(n^{\omega+1})$ many arithmetic operations.
\end{proof}
\subsection{Computing the complexity of the randomized algorithm in \cite{koiran2020derandomization} and comparing it with our algorithm :}
Recall the \cite{koiran2020derandomization} equivalence from Section \ref{sec:cubes}. The algorithm proceeds as follows:
\begin{enumerate}
    \item Pick a random matrix $R \in M_n (\mathbb{K})$ and set $h(x) = f(Rx)$.
    \item Let $T_1 ,..., T_n$ be the slices of $h$. If $T_1$ is singular, reject. Otherwise,
compute $T_1' = T_1^{-1}$ .
    \item If the matrices $T_1'T_k$ commute and are all diagonalisable over $\mathbb{K}$, accept. Otherwise, reject.
\end{enumerate}
The following are the different stages of computation required in this algorithm:
\begin{enumerate}
    \item   If the polynomial is input in dense representation, here we have to compute all the slices $T_1,...,T_n$ and following the proof of Theorem~\ref{thm:complexitycubes}, this takes $O(n^4)$ many arithmetic operations
    \newline
    If the polynomial is given as a blackbox, we compute $x' = Rx$. And we call the blackbox on this input. 
    \item \textbf{Compute } $T_{1}, T_{2}, ... , T_{n}$
        We know 
        \begin{align*}
            (T_{k})_{ij} = \frac{1}{3!}\partial_{x_ix_jx_k}(h). 
        \end{align*}
         So we can extract each entry of $T_{k}$ using constant many calls to the blackbox and constantly many arithemtic operations using Lemma~ \ref{lem:circuitcomplexity}. There are in total $n^3$ such entries that we need to compute. So the total number of calls to the blackbox is $O(n^3)$ and the number of arithmetic operations is $O(n^3)$
    \item \textbf{Check if }$T_1$\textbf{ is invertible. If invertible, compute }
    $T_1' = T_1^{-1}$.
    \newline
    This can be done in time at most $O(n^3)$.
    \item \textbf{Checking pairwise commutativity of }$\{T_1'T_j\}_{j \in [n]}$ 
    \newline
    Here we compute the product $T_1'T_iT_1'T_j$ and $T_1'T_jT_1'T_i$ and check if their difference is $0$. For each pair, this can be done in time $O(n^{\omega})$. Since there are $\binom{n}{2}$ many pairs, this can be done in time $O(n^{\omega+2})$.
    \item \textbf{Checking the diagonalisability of }$T_1'T_j$ \textbf{ for all } $j \in [n]$:
    \newline
    As we showed that the diagonalisability of each $T_1'T_j$ can be checked in time $O(n^{\omega+1})$. So the total time taken for checking diagonalisability of $n$ such matrices is $O(n^{\omega+2})$.
\end{enumerate}
So we conclude that if the polynomial is given as an input in the dense representation model, then the algorithm runs in time $O(n^{\omega+2})$.
\newline
If the polynomial is given as a blackbox, then the algorithm makes $O(n^3)$ many calls to the blackbox and requires $O(n^{\omega+2})$ many arithmetic operations.
\par
So we manage to shave a factor of~$n$ from \cite{koiran2020derandomization} in both cases: when the polynomial is given in dense representation as well as when it is input as an arithmetic circuit.
\section{Complexity analysis for equivalence to some polynomial in $\mathcal{P}_d$}\label{app:dthpowers}
\subsection{Complexity Analysis in the algebraic model}\label{app:pdalgebraic}
In this section, we provide a detailed complexity analysis of Algorithm \ref{alg:pd}. We show that if a degree $d$ polynomial in $n$ variables over $\mathbb{C}$ is given as a blackbox, the algorithm makes $\text{poly}(n,d)$ many calls to the blackbox and performs $\text{poly}(n,d)$ many arithmetic operations to decide if $f$ is equivalent to some polynomial in $\mathcal{P}_d$.
\begin{theorem}\label{thm:pdcomplexity}
If a degree-d form $f \in \mathbb{C}[x_1,...,x_n]$ is given as a blackbox, then Algorithm \ref{alg:pd} makes $O(n^2d)$ many calls to the blackbox and requires $O(n^2d \log^2d \log\log d + n^{\omega+1})$ many arithmetic operations.
\end{theorem}
\begin{proof}
The following are the different stages of computation required in this algorithm:
\begin{enumerate}
    \item  If the polynomial is given as a blackbox, we compute $x' = Rx$. And we call the blackbox on this input.  
    
    \item Compute $T_{\bar{1}}, T_{\bar{2}}, T_{\bar{3}}$. We know that $(T_{\bar{k}})_{ij} = \frac{1}{d!}\partial_{x_ix_jx_k^{d-2}}(h)$.So we can extract each entry of $T_{\bar{i}}$ using $O(d)$ many oracle calls to $C'$ and $O(M(d)\log d)$ many arithmetic operations using Lemma (\ref{lem:circuitcomplexity}). There are in total $3n^2$ such entries that we need to compute. So this entire operation can be done using $O(n^2d)$ many oracle calls to the blackbox and $O(n^2M(d)\log d)$ many arithmetic operations.
    \item \textbf{Check if }$T_{\bar{1}}$\textbf{ is invertible. If invertible, compute }
        $T'_{\bar{1}}= (T_{\bar{1}})^{-1}$
%    \newline
%    This can be done in time at most $O(n^3)$
    \item \textbf{Checking commutativity of }$T'_{\bar{1}}T_{\bar{2}}$ and $T'_{\bar{1}}T_{\bar{3}}$.
%    \newline
%    Here we compute the product $T'_{\bar{1}}T_{\bar{2}}T'_{\bar{1}}T_{\bar{3}}$ and $T'_{\bar{1}}T_{\bar{3}}T'_{\bar{1}}T_{\bar{2}}$ and check if their difference is $0$. This can be done in time $O(n^{\omega})$. 
    \item \textbf{Checking the diagonalisability of }$T'_{\bar{1}}T_{\bar{2}}$:
%\newline
%Here we use Lemma (\ref{lem:compdiag}). Hence there are four steps:
%\begin{itemize}
%    \item Compute the characteristic polynomial of $M$ i.e. $\chi_M$. This can be done in time $O(n^{\omega}\log n)$ (\cite{KELLERGEHRIG1985309})
%    \item Compute gcd$(\chi_M,\chi_M',)$. Since, $deg(\chi_M) \leq n$, this can be done in $O(n^2)$ using Euclidean Algorithm. (\cite{10.5555/2512973})
%    \item Compute $P_M = \frac{\chi_M}{\text{gcd}(\chi_M,\chi_M' )}$. This is can be computed in $O(n^2)$ using the standard long-division algorithm.
%    \item Check if $P_M(M) = 0$. Using Horner's Method, we can evaluate the polynomial at $M$ using $n$ many matrix multiplications only. Hence, computing $P_M(M)$ takes $O(n^{\omega+1})$ time. 
%\end{itemize}
%Hence, we can conclude that the diagonalisability of $T_{\bar{1}}'T_{\bar{2}}$ can be checked in time $O(n^{\omega+1})$
\end{enumerate}
Steps (3), (4) and (5) are exactly the same as Theorem \ref{thm:complexitycubes}. This is because they don't require any assumptions on $T_{\bar{1}},T_{\bar{2}},T_{\bar{3}}$ except for the fact that they are $n \times n$ matrices. Thus, from the proof of Theorem \ref{thm:complexitycubes}, we get that these steps can be checked in $O(n^{\omega+1})$ many arithmetic operations.
\par
So we conclude that if the polynomial is given as a blackbox, then the algorithm makes $O(n^2d)$ many calls to the blackbox; the number of arithmetic operations required is $O(n^2M(d)\log d + n^{\omega+1})$.
\newline
Since $M(d) = O(d\log d \log\log d)$, the number of arithmetic operations is $O(n^2d \log^2d \log\log d + n^{\omega+1})$.
\subsection{Complexity analysis for the bit model}\label{sec:bitcomplexity}
We look at the case where $f \in \mathbb{Q}[x_1,...,x_n]$ is a degree-$d$ form and we want to check if it can be written as a linear combination of $d$th powers of linear forms over $\mathbb{R}$ or $\mathbb{C}$. Our algebraic algorithms run in polynomial time in the standard bit model of computation, i.e., they are “strongly polynomial” algorithms (this is not automatic due to the issue of coefficient growth during the computation). For a detailed discussion of how the previous algorithms fail to give a polynomial time algorithm for this problem in this model, refer to Section 1.1 from \cite{koiran2020derandomization}. 
\par
We try to estimate the complexity of each step of the algorithm in the standard bit model of computation, following the proof of Theorem \ref{thm:pdcomplexity}. In Step (1), we take a matrix $R$ such that its entries $r_{i,j}$ are picked uniformly and independently at random from a finite set $S$.Hence the bit size of the entries of $R$ are bounded by $\log(|S|) + 1$. We define $h = f(Rx)$.  Recall from Theorem \ref{thm:Pdslicestucture}, that the slices $T_{\bar{i}}$ of $h$ can be written as $R^T(\sum_{i_1,...,i_{d-2} \in [n]} (\prod_{m \in [d]}r_{i_m,i})S_{i_1...i_{d-2}})R$. The entries of the slices $S_{i_1...i_{d-2}}$ are essentially the coefficients of $f$. So they are bounded by the bit size of the maximum coefficient in $f$ which we define to be $b_f$. Therefore, the bit size of each element of $T_{\bar{i}}$ is $b := \text{poly}(\log(|S|),\log(n),d,b_f)$ for all $i$. Now the elements of $T_{\bar{1}},T_{\bar{2}},T_{\bar{3}}$ are computed using Lemma \ref{lem:circuitcomplexity} that uses polynomial interpolation and hence, computing these matrices takes time $\text{poly}(n,d,b)$. In Step (3), we check if the slice $T_{\bar{1}}$ is invertible and if invertible, it is inverted. Since the bit-size of the inputs of $T_{\bar{1}}$ are bounded by $b$, the matrix can be inverted using Bareiss' Algorithm \cite{Bar68} in time $\text{poly}(n,d,b)$. In Step (4), testing commutativity of $T_{\bar{1}}'T_{\bar{2}}$ and $T_{\bar{1}}'T_{\bar{3}}$ requires only matrix multiplication which does not blow up the entry of the matrices and hence, this step can be done in time $\text{poly}(n,d,b)$. In Step (5), we need to check the diagonalisability of $M = T_{\bar{1}}'T_{\bar{2}}$. Over the field of complex numbers it therefore suffices to check that $P_M(M) = 0$ which can be done in time $\text{poly}(n,d,b)$. For the discussion of how the same can be executed over $\mathbb{R}$, refer to Section 4 of \cite{koiran2020derandomization}.
\par
So the total time required for the entire computation in the bit model of complexity is $\text{poly}(n,d,log(|S|),b_f)$ where $b_f$ is the bit size of the maximum coefficient in $f$ and $S$ is the set from which the entries of $R$ are picked uniformly and independently at random. 
%%%%%%%%%%%%%%%
%we need to check additionally that all the roots of $P_M$
%are real. This can be done for instance with the help of Sturm sequences,
%which can be used to compute the number of roots of a real polynomial on
%any real (possibly unbounded) interval. Alternatively, the number of real roots of a real polynomial can be obtained through Hurwitz determinants
%(\cite{RS02}, Corollary 10.6.12), and is given by the signature of the Hermite quadratic form (\cite{Basu06}, Theorem 4.48). The arithmetic cost of these methods is polynomially
%bounded, and they can also be implemented to run in polynomial time in
%the bit model. \footnote{For Sturm sequences this is not obvious because in a naive implementation, the bit size of the numbers involved may grow exponentially. There is however an efficient implementation based on subresultants \cite{Basu06}. The same issue of coefficient growth already occurs in the computation of the gcd of two polynomials, and can also be solved with subresultants}
%%%%%%%%%%%%%%%%%%%%%%%%%%%%%%%%%%%%%%%%%%
\end{proof}

\section{Complexity analysis for variable minimization}\label{app:minvars}
In this section, we provide a detailed complexity analysis of Algorithm \ref{alg:minvars}. We show that if a degree $d$ polynomial in $n$ variables over $\mathbb{C}$ is given as a blackbox, the algorithm performs $\text{poly}(n,d)$ many calls to the blackbox and performs $\text{poly}(n,d)$ many arithmetic operations to check if the polynomial can be written as a linear combination of $t$ many linearly independent linear forms for some $t \leq n$ (see Theorem \ref{thm:minvarscomp}).
\par
In the next lemma, we show that given a blackbox computing the polynomial, we can compute the partial derivative with respect to a single variable at a given point in $O(d)$ many oracle calls and $\text{poly}(d)$ arithmetic operations.
\begin{lemma}\label{lem:pardermatrixcomp}
Given a blackbox computing $P \in \mathbb{K}[x_1,...,x_n]$ of degree at most $ d$, and given points $\alpha_1,...,\alpha_n \in \mathbb{K}^n$, we can  compute the matrix $M$ such that  $M_{ij} = \partial_{x_i}(P)(\alpha_j)$ using $O(n^2d)$ many oracle calls to the blackbox and $O(n^2M(d)\log d)$ many arithmetic operations.
\end{lemma}
\begin{proof}
We assume $i=1$.
Given points $a_1,...,a_n \in \mathbb{K}$, $\partial_{x_1}(P)(a_1,...,a_n)$ can be computed using the following steps of computation:
\begin{itemize}
    \item Compute the polynomial $P(x,a_2,...,a_n)$ explicitly using polynomial interpolation. This can be done using $O(M(d)\log d)$ many arithmetic operations and $d+1$ many calls to the blackbox. \cite{10.5555/2512973} (Section 10.2)
    \item Compute $\partial_x(P(x,a_2,...,a_n))$ using $O(d)$ many arithmetic operations.
    \item evaluate it at $x = a_1$, which requires $O(d)$ many arithmetic operations.
\end{itemize}
We do this for every $(a_1,...,a_n) = (\alpha_j^{(1)},...,\alpha_j^{(n)})$ for all $j \in [n]$. Thus each row of the matrix $M$ can be computed using $O(nM(d)\log d)$ many operations. For each $x_i$ such that $i \in [n]$, a similar algorithm works and the matrix $M$ can be computed using $O(n^2M(d)\log d)$ arithmetic operations and $O(n^2d)$ many calls to the blackbox.
\end{proof}
Now we are in a position to analyse the complexity of Algorithm \ref{alg:minvars}
\begin{theorem}\label{thm:minvarscomp}
Given a blackbox computing a polynomial $P \in \mathbb{K}[x_1,...,x_n]_d$ with $t$ essential variables, the algorithm makes $O(n^2d)$ many calls to the blackbox and 
requires $O(n^{\omega}+ n^2M(d)\log d + t^{\omega+1})$ many arithmetic operations.
\end{theorem}
\begin{proof}
The following are the steps of the algorithm:
\begin{itemize}
    \item Using Lemma \ref{lem:pardermatrixcomp}, we know that the matrix $M$ can be computed in time $O(n^2M(d)\log d)$ many arithmetic operations and $O(n^2d)$ many calls to the blackbox.
    \item The basis of $\text{ker}(M)$ can be computed and the basis can be completed using $O(n^{\omega})$ many arithmetic operations.
    \item For $n = t$, using Theorem \ref{thm:pdcomplexity} the algorithm checks equivalence using $O(t^2d)$ many calls to the blackbox and $O(t^2M(d)\log d + t^{\omega+1})$ many arithmetic operations.
\end{itemize}
So we conclude that if the polynomial is given as a blackbox, then the algorithm makes $O(n^2d)$ many calls to the blackbox and the number of arithmetic operations required is $O(n^{\omega}+ n^2M(d)\log d + t^{\omega+1})$.
\par
\end{proof}
Using Section \ref{sec:bitcomplexity} and the fact that we just need to do polynomial interpolation and completion of the basis, we conclude that Algorithm \ref{alg:minvars} takes time $\text{poly}(n,d,\log(|S|),b_f)$ in the bit model of computation. Here $b_f$ is the bit size of the maximum coefficient in $f$ and $S$ is the set from which the entries of $R$ and $\alpha_j^{(i)}$s are picked uniformly and independently at random.
\section*{Acknowledgements}
We would like to thank Mateusz Skomra for useful discussions in the early stages of this work, and Frédéric Magniez for discussions on commutativity testing.

%\bibliographystyle{alpha}
%\bibliography{KoiranSkomra,other}

\begin{thebibliography}{BGVKS20}

\bibitem[AGH{\etalchar{+}}14]{JMLR:v15:anandkumar14b}
Animashree Anandkumar, Rong Ge, Daniel Hsu, Sham~M. Kakade, and Matus
  Telgarsky.
\newblock Tensor decompositions for learning latent variable models.
\newblock {\em Journal of Machine Learning Research}, 15(80):2773--2832, 2014.

\bibitem[Bar68]{Bar68}
Erwin~H. Bareiss.
\newblock Sylvester's identity and multistep integer-preserving gaussian
  elimination.
\newblock {\em Mathematics of Computation}, 22(103):565--578, 1968.

\bibitem[BCMT10]{brachat10}
Jérôme Brachat, Pierre Comon, Bernard Mourrain, and Elias Tsigaridas.
\newblock Symmetric tensor decomposition.
\newblock {\em Linear Algebra and its Applications}, 433(11-12):1851--1872,
  2010.

\bibitem[BCSS98]{BCSS}
L.~Blum, F.~Cucker, M.~Shub, and S.~Smale.
\newblock {\em Complexity and Real Computation}.
\newblock Springer-Verlag, 1998.

\bibitem[BGI11]{bernardi11}
Alessandra Bernardi, Alessandro Gimigliano, and Monica Ida.
\newblock Computing symmetric rank for symmetric tensors.
\newblock {\em Journal of Symbolic Computation}, 46(1):34--53, 2011.

\bibitem[BGVKS20]{BGKS20}
Jess Banks, Jorge Garza-Vargas, Archit Kulkarni, and Nikhil Srivastava.
\newblock Pseudospectral shattering, the sign function, and diagonalization in
  nearly matrix multiplication time.
\newblock In {\em 2020 IEEE 61st Annual Symposium on Foundations of Computer
  Science (FOCS)}, pages 529--540, 2020.

\bibitem[BSS89]{BSS89}
L.~Blum, M.~Shub, and S.~Smale.
\newblock On a theory of computation and complexity over the real numbers:
  {NP}-completeness, recursive functions and universal machines.
\newblock {\em Bulletin of the American Mathematical Society}, 21(1):1--46,
  July 1989.

\bibitem[BSV21]{bhargava21}
Vishwas Bhargava, Shubhangi Saraf, and Ilya Volkovich.
\newblock Reconstruction algorithms for low-rank tensors and depth-3
  multilinear circuits.
\newblock {\em arXiv preprint arXiv:2105.01751}, 2021.

\bibitem[Car06a]{carlini06}
Enrico Carlini.
\newblock Reducing the number of variables of a polynomial.
\newblock In {\em Algebraic geometry and geometric modeling}, Math. Vis., pages
  237--247. Springer, Berlin, 2006.

\bibitem[Car06b]{10.1007/978-3-540-33275-6_15}
Enrico Carlini.
\newblock Reducing the number of variables of a polynomial.
\newblock In Mohamed Elkadi, Bernard Mourrain, and Ragni Piene, editors, {\em
  Algebraic Geometry and Geometric Modeling}, pages 237--247, Berlin,
  Heidelberg, 2006. Springer Berlin Heidelberg.

\bibitem[CG05]{cheze05}
Guillaume Cheze and Andr{\'e} Galligo.
\newblock Four lectures on polynomial absolute factorization.
\newblock In {\em Solving polynomial equations}, pages 339--392. Springer,
  2005.

\bibitem[CL07]{ChezeLecerf07}
Guillaume Ch{\`e}ze and Gr{\'e}goire Lecerf.
\newblock Lifting and recombination techniques for absolute factorization.
\newblock {\em Journal of Complexity}, 23(3):380--420, 2007.

\bibitem[DL78]{DL78}
Richard~A. Demillo and Richard~J. Lipton.
\newblock A probabilistic remark on algebraic program testing.
\newblock {\em Information Processing Letters}, 7(4):193--195, 1978.

\bibitem[FGS18]{forbes_et_al:LIPIcs:2018:9058}
Michael~A. Forbes, Sumanta Ghosh, and Nitin Saxena.
\newblock {Towards Blackbox Identity Testing of Log-Variate Circuits}.
\newblock In {\em 45th International Colloquium on Automata, Languages, and
  Programming (ICALP 2018)}, 2018.

\bibitem[Fre79]{freivalds79}
R{\=u}si{\c{n}}{\v{s}} Freivalds.
\newblock Fast probabilistic algorithms.
\newblock In {\em International Symposium on Mathematical Foundations of
  Computer Science (MFCS)}, pages 57--69. Springer, 1979.

\bibitem[Gao03]{gao03}
Shuhong Gao.
\newblock Factoring multivariate polynomials via partial differential
  equations.
\newblock {\em Mathematics of computation}, 72(242):801--822, 2003.

\bibitem[GG13]{10.5555/2512973}
Joachim von~zur Gathen and J{\"u}rgen Gerhard.
\newblock {\em Modern Computer Algebra}.
\newblock Cambridge University Press, USA, 3rd edition, 2013.

\bibitem[GGKS19]{garg19}
Ankit Garg, Nikhil Gupta, Neeraj Kayal, and Chandan Saha.
\newblock Determinant equivalence test over finite fields and over
  $\mathbb{Q}$.
\newblock In {\em Electronic Colloquium on Computational Complexity (ECCC)},
  volume~26, page~42, 2019.

\bibitem[GMKP17]{GKP17ISSAC}
Ignacio Garc\'{\i}a-Marco, Pascal Koiran, and Timoth\'ee Pecatte.
\newblock Reconstruction algorithms for sums of affine powers.
\newblock In {\em Proc. International Symposium on Symbolic and Algebraic
  Computation (ISSAC)}, pages 317--324, 2017.

\bibitem[GMKP18]{GKP18}
Ignacio Garc\'ia-Marco, Pascal Koiran, and Timoth{\'e}e Pecatte.
\newblock Polynomial equivalence problems for sums of affine powers.
\newblock In {\em Proc. International Symposium on Symbolic and Algebraic
  Computation (ISSAC)}, 2018.

\bibitem[Har70]{Harshman1970FoundationsOT}
R.~Harshman.
\newblock Foundations of the parafac procedure: Models and conditions for an
  "explanatory" multi-model factor analysis.
\newblock 1970.

\bibitem[HJ13]{horn13}
Roger Horn and Charles Johnson.
\newblock {\em Matrix Analysis}.
\newblock Cambridge University Press (second edition), 2013.

\bibitem[Kay11]{kayal11}
Neeraj Kayal.
\newblock Efficient algorithms for some special cases of the polynomial
  equivalence problem.
\newblock In {\em Symposium on Discrete Algorithms (SODA)}. Society for
  Industrial and Applied Mathematics, January 2011.

\bibitem[KB09]{Kolda2009TensorDA}
T.~Kolda and B.~Bader.
\newblock Tensor decompositions and applications.
\newblock {\em SIAM Rev.}, 51:455--500, 2009.

\bibitem[KG85]{KELLERGEHRIG1985309}
Walter Keller-Gehrig.
\newblock Fast algorithms for the characteristics polynomial.
\newblock {\em Theoretical Computer Science}, 36:309 -- 317, 1985.

\bibitem[KNST18]{kayal18}
Neeraj Kayal, Vineet Nair, Chandan Saha, and S{\'e}bastien Tavenas.
\newblock Reconstruction of full rank algebraic branching programs.
\newblock {\em {ACM} Transactions on Computation Theory ({TOCT})}, 11(1):2,
  2018.

\bibitem[KS09]{karnin09}
Zohar Karnin and Amir Shpilka.
\newblock Reconstruction of generalized depth-3 arithmetic circuits with
  bounded top fan-in.
\newblock In {\em 24th Annual IEEE Conference on Computational Complexity
  (CCC)}, pages 274--285, 2009.

\bibitem[KS19]{kayal19}
Neeraj Kayal and Chandan Saha.
\newblock Reconstruction of non-degenerate homogeneous depth three circuits.
\newblock In {\em Proc. 51st Annual {ACM} Symposium on Theory of Computing
  (STOC)}, pages 413--424, 2019.

\bibitem[KS20]{koiran2020derandomization}
Pascal Koiran and Mateusz Skomra.
\newblock Derandomization and absolute reconstruction for sums of powers of
  linear forms.
\newblock {\em arXiv preprint arXiv:1912.02021}, 2020.
\newblock To appear in Theoretical Computer Science, 2021.

\bibitem[MN07]{magniez07}
Fr{\'e}d{\'e}ric Magniez and Ashwin Nayak.
\newblock Quantum complexity of testing group commutativity.
\newblock {\em Algorithmica}, 48(3):221--232, 2007.

\bibitem[Moi18]{moitra2018algorithmic}
A.~Moitra.
\newblock {\em Algorithmic Aspects of Machine Learning}.
\newblock Cambridge University Press, 2018.

\bibitem[Pak12]{Pak12}
Igor Pak.
\newblock Testing commutativity of a group and the power of randomization.
\newblock {\em LMS Journal of Computation and Mathematics}, 15:38--43, 2012.

\bibitem[RS00]{schulman00}
Sridhar Rajagopalan and Leonard~J Schulman.
\newblock Verification of identities.
\newblock {\em SIAM Journal on Computing}, 29(4):1155--1163, 2000.

\bibitem[Sch80]{Sch80}
J.~T. Schwartz.
\newblock Fast probabilistic algorithms for verification of polynomial
  identities.
\newblock {\em J. ACM}, 27(4):701–717, October 1980.

\bibitem[Sha09]{shaker09}
Hani Shaker.
\newblock Topology and factorization of polynomials.
\newblock {\em Mathematica Scandinavica}, pages 51--59, 2009.

\bibitem[Shi16]{shitov16}
Yaroslav Shitov.
\newblock How hard is the tensor rank?
\newblock {\em arXiv preprint arXiv:1611.01559}, 2016.

\bibitem[Shp09]{shpilka09}
Amir Shpilka.
\newblock Interpolation of depth-3 arithmetic circuits with two multiplication
  gates.
\newblock {\em {SIAM} Journal on Computing}, 38(6):2130--2161, 2009.

\bibitem[SS71]{Schonhage1971}
A.~Sch{\"o}nhage and V.~Strassen.
\newblock Schnelle multiplikation gro{\ss}er zahlen.
\newblock {\em Computing}, 7(3):281--292, Sep 1971.

\bibitem[Zip79]{Zip79}
Richard Zippel.
\newblock Probabilistic algorithms for sparse polynomials.
\newblock In {\em Symbolic and Algebraic Computation}, pages 216--226, Berlin,
  Heidelberg, 1979. Springer Berlin Heidelberg.

\end{thebibliography}

\newcommand{\etalchar}[1]{$^{#1}$}

%\begin{comment}
%\newpage
%\section{(Wishful) Theorem for extending from $P_3$ to $P_4$}
%\subfile{extendingtheorem}
%\end{comment}

\end{document}